\documentclass[11pt]{amsart}
\usepackage{amsmath,amsfonts,amssymb,amsthm,mathscinet,enumitem}
\usepackage{hyperref}
\usepackage[noabbrev,capitalize]{cleveref}

\theoremstyle{plain}
\newtheorem{theorem}{Theorem}[section]
\newtheorem{proposition}[theorem]{Proposition}
\newtheorem{lemma}[theorem]{Lemma}
\newtheorem{corollary}[theorem]{Corollary}
\theoremstyle{definition}
\newtheorem{definition}[theorem]{Definition}
\newtheorem{notation}[theorem]{Notation} 
\theoremstyle{remark}
\newtheorem{remark}[theorem]{Remark}

\newcommand{\lindexset}{{[2]\times[2]}}
\newcommand{\indexset}{{[p]\times\lindexset}}
\newcommand{\Eindexset}{{E\times\lindexset}}
\newcommand{\Findexset}{{F\times\lindexset}}
\newcommand{\eindexset}{{\{e\}\times\lindexset}}
\newcommand{\findexset}{{\{f\}\times\lindexset}}
\newcommand{\gindexset}{{\{g\}\times\lindexset}}
\newcommand{\hindexset}{{\{h\}\times\lindexset}}
\newcommand{\eeindexset}{{\{\epsilon(e)\}\times\lindexset}}

\renewcommand{\epsilon}{\varepsilon}
\renewcommand{\phi}{\varphi}
\DeclareMathOperator{\ADF}{ADF}

\DeclareMathOperator{\Con}{Con}
\DeclareMathOperator{\Sep}{Sep}
\DeclareMathOperator{\Seq}{Seq}
\DeclareMathOperator{\gelo}{GELO}
\DeclareMathOperator{\Part}{Part}
\DeclareMathOperator{\ssac}{SSAC}
\DeclareMathOperator{\As}{As}
\DeclareMathOperator{\Sat}{Sat}
\DeclareMathOperator{\Wr}{Wr}
\DeclareMathOperator{\Isom}{Isom}
\DeclareMathOperator{\Sols}{Sols}
\DeclareMathOperator{\Span}{span}
\DeclareMathOperator{\Stab}{Stab}
\DeclareMathOperator{\Co}{Co}
\DeclareMathOperator{\Fo}{Fo}
\newcommand{\mso}{[\kern-1.75pt [}
\newcommand{\msc}{]\kern-1.75pt ]}
\newcommand{\ms}[1]{\mso #1 \msc}
\newcommand{\E}{{\mathbb E}}
\newcommand{\N}{{\mathbb N}}
\newcommand{\Q}{{\mathbb Q}}
\newcommand{\Z}{{\mathbb Z}}
\newcommand{\cC}{{\mathcal C}}
\newcommand{\cD}{{\mathcal D}}
\newcommand{\cM}{{\mathcal M}}
\newcommand{\cP}{{\mathcal P}}
\newcommand{\cQ}{{\mathcal Q}}
\newcommand{\cR}{{\mathcal R}}
\newcommand{\cS}{{\mathcal S}}
\newcommand{\cT}{{\mathcal T}}
\newcommand{\cU}{{\mathcal U}}
\newcommand{\cV}{{\mathcal V}}
\newcommand{\cW}{{\mathcal W}}
\newcommand{\cWp}{\cW^{(p)}}
\newcommand{\fP}{{\mathfrak P}}
\newcommand{\fQ}{{\mathfrak Q}}
\newcommand{\card}[1]{\left|{#1}\right|}
\newcommand{\conj}[1]{\overline{#1}}
\newcommand{\expv}[1]{{\mathbf E}_{#1}}
\newcommand{\ev}{\E_f^\ell}
\newcommand{\mom}[1]{\mu_{#1,f}^\ell}
\newcommand{\smom}[1]{\tilde{\mu}_{#1,f}^\ell}

\newcommand{\sums}[1]{\sum_{\substack{#1}}}
\newcommand{\ceil}[1]{\lceil{#1}\rceil}

\title[Limiting Moments of Autocorrelation Demerit Factors]{Limiting Moments of Autocorrelation Demerit Factors of Binary Sequences}

\author{Daniel J.\ Katz}
\address{Department of Mathematics, California State University, Northridge, \: United States}
\author{Miriam E.\ Ramirez}
\thanks{Daniel J.\ Katz is with the Department of Mathematics, California State University, Northridge.  Miriam E.\ Ramirez was with the Department of Mathematics, California State University, Northridge, USA.  This paper is based upon work of both authors supported in part by the National Science Foundation under Grants DMS-1500856 and CCF-1815487, and by work of Daniel J.\ Katz supported in part by the National Science Foundation under Grant CCF-2206454.}
\date{24 October 2024}

\begin{document}

\begin{abstract}
Various problems in engineering and natural science demand binary sequences that do not resemble translates of themselves, that is, the sequences must have small aperiodic autocorrelation at every nonzero shift.
If $f$ is a sequence, then the demerit factor of $f$ is the sum of the squared magnitudes of the autocorrelations at all nonzero shifts for the sequence obtained by normalizing $f$ to unit Euclidean norm.
The demerit factor is the reciprocal of Golay's merit factor, and low demerit factor indicates low self-similarity of a sequence under translation.
We endow the $2^\ell$ binary sequences of length $\ell$ with uniform probability measure and consider the distribution of their demerit factors.
Earlier works used combinatorial techniques to find exact formulas for the mean, variance, skewness, and kurtosis of the distribution as a function of $\ell$.
These revealed that for $\ell \geq 4$, the $p$th central moment of this distribution is strictly positive for every $p \geq 2$.
This article shows that for every $p$, the $p$th central moment is $\ell^{-2 p}$ times a quasi-polynomial function of $\ell$ with rational coefficients.
It also shows that, in the limit as $\ell$ tends to infinity, the $p$th standardized moment is the same as that of the standard normal distribution.  
\end{abstract}

\maketitle

\section{Introduction}

This paper is concerned with the aperiodic autocorrelation of binary sequences, which is important in communications and ranging applications for establishing accurate timings; see \cite{Golomb,Golomb-Gong,Schroeder}.
Because we are viewing sequences aperiodically, we define a {\it binary sequence of length $\ell$} to be of the form $f=(\ldots,f_{-2},f_{-1},f_0,f_1,f_2,\ldots)$, where $f_0,f_1,\ldots,f_{\ell-1} \in \{-1,1\}$ and $f_j=0$ for all other $j \in \Z$.
This definition makes it easy to compare the sequence with non-cyclic translates of itself.
The {\it aperiodic autocorrelation at shift $s$} of the sequence $f$ is defined to be
\[
C_f(s) = \sum_{j \in \Z} f_{j+s} \conj{f_j}.
\]
This is always a finite value because $f_j\not=0$ for only finitely many $j$, which also guarantees that $C_f(s)=0$ for all but finitely many $s$.
Note that $C_f(0)$ is the squared Euclidean norm of $f$, so it is always a nonnegative real number (and strictly positive when $f\not=0$).

When a sequence $f$ is used in applications, it is desirable that $|C_f(s)|$ be small compared to $C_f(0)$ for all nonzero shifts $s$; this helps establish correct timing when the sequence is used for synchronization.
One measure of lowness of autocorrelation of $f$ is the {\it peak sidelobe level of $f$}, which is $\max_{s \in \Z\smallsetminus\{0\}} |C_f(s)|$; this measures the worst case (highest) autocorrelation value off the main lobe (shift $0$).
A mean square measure of lowness of autocorrelation is the demerit factor of $f$.
The {\it (autocorrelation) demerit factor of $f$} is the sum of the squared magnitudes of the autocorrelation at every nonzero shift $s$ of the sequence one obtains by normalizing $f$ to unit Euclidean norm.
That is,
\begin{equation}\label{Stan}
\ADF(f)= \frac{\sum_{s \in \Z \smallsetminus\{0\}} |C_f(s)|^2}{C_f(0)^2} = -1 + \frac{\sum_{s \in \Z} |C_f(s)|^2}{C_f(0)^2}.
\end{equation}
Golay's merit factor \cite{Golay-75} is the reciprocal of the demerit factor, so favorable sequences have a low demerit factor and a high merit factor.
We often want to study only the numerator of the last fraction in \eqref{Stan}, which is the sum of the squared magnitudes of all the autocorrelation values, so we define
\[
\ssac(f)=\sum_{s \in \Z} |C_f(s)|^2,
\]
so that
\begin{equation}\label{Gabrielle}
\ADF(f)=-1+\frac{\ssac(f)}{C_f(0)^2}.
\end{equation}

For each positive integer $\ell$, we are interested in the distribution of demerit factors of the set  $\Seq(\ell)$ of the $2^\ell$ binary sequences of length $\ell$, endowed with uniform probability distribution.
The expected value of a random variable $v$ with respect to this distribution is denoted $\ev v(f)=\expv{f \in \Seq(\ell)}(v(f))$.
Then the $p$th central moment is denoted
\[
\mom{p} v(f) = \ev \left(v(f)-\ev v(f)\right)^p,
\]
and the $p$th standardized moment is
\[
\smom{p} v(f) = \frac{\mom{p} v(f)}{\left(\mom{2} v(f)\right)^{p/2}}.
\]
Since $C_f(0)=\ell$ for every binary sequence $f$ of length $\ell$, we have $\ADF(f)=-1+\ssac(f)/\ell^2$ by \eqref{Gabrielle}, so that
\begin{equation}\label{Natasha}
\mom{p} \ADF(f) = \frac{\mom{p}\ssac(f)}{\ell^{2 p}},
\end{equation}
and of course $\smom{p} \ADF(f) = \smom{p} \ssac(f)$.

Now we discuss previous results on moments of the autocorrelation demerit factor.
The mean demerit factor was determined by Sarwate \cite[eq.\ (13)]{Sarwate}.
\begin{theorem}[Sarwate, 1984]
For every positive integer $\ell$, we have
\[
\ev \ADF(f) = 1-\frac{1}{\ell}.
\]
\end{theorem}
The variance was explicitly determined by Jedwab \cite[Thm.\ 1]{Jedwab}.
\begin{theorem}[Jedwab, 2024]\label{Albert}
For every positive integer $\ell$, we have
\[
\mom{2} \ADF(f) = \begin{cases}
\frac{16 \ell^3-60 \ell^2+56 \ell}{3 \ell^4} & \text{if $\ell$ is even,} \\
\frac{16 \ell^3-60 \ell^2+56 \ell-12}{3 \ell^4} & \text{if $\ell$ is odd.}
\end{cases}
\]
\end{theorem}
Note that this shows that the variance tends to $0$ as $\ell$ tends to infinity, a fact earlier proved by Borwein and Lockhart \cite[pp.\ 1469--1470]{Borwein-Lockhart}.
In \cite[Thm.\ 1.3]{Katz-Ramirez}, a combinatorial theory was developed for determining the $p$th central and standardized moments of the demerit factor, which enabled the explicit calculation of the skewness 
\begin{theorem}[Katz--Ramirez, 2024]\label{Barbara}
For every positive integer $\ell$, we have 
\[
\mom{3} \ADF(f)=
\begin{cases}
\frac{160\ell^4-1296\ell^3+3296\ell^2-2496\ell}{\ell^6}, &  \text{if $\ell \equiv 0 \bmod 4$,} \\
\frac{160\ell^4-1296\ell^3+3296\ell^2-2736\ell+576}{\ell^6}, &  \text{if $\ell \equiv \pm 1 \bmod 4$,} \\ 
\frac{160\ell^4-1296\ell^3+3296\ell^2-2496\ell-384}{\ell^6}, &  \text{if $\ell \equiv 2 \bmod 4$,}
\end{cases}
\]
and
\[
\smom{3} \ADF(f)=
\begin{cases}
\frac{6\sqrt{3}(10\ell^4-81 \ell^3+206 \ell^2-156 \ell)}{(4 \ell^3-15 \ell^2+14 \ell)^{3/2}}, &  \text{if $\ell \equiv 0 \bmod 4$,} \\
\frac{6\sqrt{3}(10\ell^4-81 \ell^3+206 \ell^2-171 \ell+36)}{(4 \ell^3-15 \ell^2+14 \ell -3)^{3/2}}, &  \text{if $\ell \equiv \pm 1 \bmod 4$,} \\
\frac{6\sqrt{3}(10\ell^4-81 \ell^3+206 \ell^2-156 \ell-24)}{(4 \ell^3-15 \ell^2+14 \ell)^{3/2}}, &  \text{if $\ell \equiv 2 \bmod 4$.}
\end{cases}
\]
\end{theorem}
The previous paper also had a computer-assisted calculation of the kurtosis of the distribution of the demerit factor \cite[Thm.\ 1.4]{Katz-Ramirez}.
Another interesting fact that arises from this combinatorial theory is that all central moments are positive, except for a few that are zero when $\ell$ is small \cite[Thm.\ 1.5]{Katz-Ramirez}.
\begin{theorem}[Katz--Ramirez, 2024]
Let $\ell$ and $p$ be positive integers.
Then $\mom{p} \ADF(f)$ is nonnegative.
If either (i) $p=1$, (ii) $p>1$ is odd and $\ell\leq 3$, or (iii) $p$ is even and $\ell\leq 2$, then $\mom{p} \ADF(f)$ is zero; otherwise it is strictly positive.
\end{theorem}

Two important questions could not be answered using only the techniques of \cite{Katz-Ramirez}.
First, one notes from Theorems \ref{Albert} and \ref{Barbara} (using \eqref{Natasha}) that the central moments $\mom{p} \ssac(f)$ for $p=2$ and $p=3$ are quasi-polynomial functions of $\ell$.
In this paper, we develop techniques that show that this is true for all $p$.
\begin{theorem}\label{Florence}
Let $p \in \N$.  Then $\mom{p} \ssac(f)$ is a quasi-polynomial function of $\ell$ whose coefficients are all rational numbers.
Therefore, $\mom{p} \ADF(f)$ is $\ell^{-2 p}$ times a quasi-polynomial function of $\ell$ whose coefficients are all rational numbers.
\end{theorem}

The second question is about the general behavior of the higher moments ($p \geq 3$).
In principle, the methods of \cite{Katz-Ramirez} can be used to furnish an exact determination of any $p$th standardized moment of the distribution of demerit factors, but the calculations become more and more lengthy and complicated as $p$ increases.
To better understand the overall behavior of the moments, one wants to understand their asymptotic behavior in the limit of large sequence length $\ell$.
One cannot determine the limiting behavior using solely the methods of \cite{Katz-Ramirez}, but this paper provides further techniques that identify the dominant contributions and enable us to calculate the limiting moments.
\begin{theorem}\label{Andy}
Let $p$ be a nonnegative integer.  Then
\[
\lim_{\ell \to \infty} \smom{p} \ADF(f) = \lim_{\ell \to \infty} \smom{p} \ssac(f) =
\begin{cases}
0 & \text{if $p$ is odd,} \\
(p-1)!! & \text{if $p$ is even.}
\end{cases}
\]
\end{theorem}
Note that the limiting standardized moments are exactly the same as those of the standard normal distribution.

The rest of this paper is organized as follows.
\cref{Celia} has preliminary material and recalls the general theory of \cite{Katz-Ramirez} that enables us to calculate central moments of the distribution of demerit factors of binary sequences.
\cref{Ulrich} introduces the new ideas of this paper that are later used to prove Theorems \ref{Florence} and \ref{Andy}.
\cref{Quentin} applies this new theory to prove \cref{Florence}.
Then we investigate the dominant terms for a generic central moment in \cref{Andrew}, and determine the asymptotic standardized moments that are presented in \cref{Andy}.

\section{Preliminaries}\label{Celia}

In this section, we review the combinatorial theory of \cite{Katz-Ramirez} for calculating the moments of the distribution of demerit factors of binary sequences.
Throughout this paper, $\N=\{0,1,2,\ldots\}$ and for each $\ell\in\N$ we write $[\ell]$ for $\{0,1,\ldots,\ell-1\}$.

We use $\ms{a_1,\ldots,a_k}$ to denote the multiset $M$ where the multiplicity of an element $a$ in $M$ equals the number of times that $a$ occurs on the list $a_1,\ldots,a_k$.
If $\{a_i\}_{i \in I}$ is a family of (not necessarily distinct) elements that is indexed by the set $I$, then $\ms{a_i: i \in I}$ denotes the multiset in which the multiplicity of an element $a$ is the number of $i \in I$ such that $a_i=a$.

If $S$ and $T$ are sets, then $T^S$ is the set of functions from $S$ into $T$.
We sometimes abbreviate a tuple by omitting enclosing parentheses and commas (e.g., $(e,s,v)$ is abbreviated as $e s v$) between elements when there is no risk of confusion.

\subsection{Partitions}

A {\it partition} $\cP$ of a set $S$ is a set of nonempty, pairwise disjoint subsets (called the {\it classes of the partition}) of $S$ whose union is $S$.
The {\it modulo $\cP$ relation} is the equivalence relation on $S$ whose equivalence classes are the sets in $\cP$, i.e., if $s,t \in S$, then $s \equiv t \pmod{\cP}$ means that $s$ and $t$ lie in the same set within $\cP$.
If $f\colon S \to U$ is a function, the {\it partition induced by $f$} is the partition $\cP$ of $S$ where $s,t \in S$ lie in the same class of $\cP$ if and only if $f(s)=f(t)$.

\begin{definition}[Type of a partition]
If $\cP$ is a partition of a set $S$, then the {\it type of $\cP$} is the multiset of cardinalities of the classes in $\cP$, i.e., $\ms{|P|: P \in \cP}$.
\end{definition}

\begin{definition}[Even partition]
A partition is said to be {\it even} if every class in it is of finite, even cardinality.
\end{definition}

\begin{notation}[$\Part(E)$, $\Part(p)$]
If $E\subseteq \N$, then $\Part(E)$ is the set of all partitions of $\Eindexset$.
If $p \in \N$, then $\Part(p)$ is shorthand for $\Part([p])$.
\end{notation}

\begin{definition}[Restriction of sets and partitions]
Let $F \subseteq E \subseteq \N$.
If $P \subseteq \Eindexset$, then {\it the restriction of $P$ to $F$}, written $P_F$, is $P\cap(\Findexset)$.
If $\cP \in \Part(E)$, then {\it the restriction of $\cP$ to $F$}, written $\cP_F$, is the partition $\{P_F: P \in \cP \text{ and } P_F\not=\emptyset\}$, which resides in $\Part(F)$.
\end{definition}

\begin{definition}[Globally even, locally odd (GELO) partition]
Let $E \subseteq \N$ and $\cP \in \Part(E)$.
Then we call $\cP$ {\it globally even, locally odd} (or just {\it GELO}) to mean that $\cP$ is an even partition such that for every $e \in E$, the restricted partition $\cP_{\{e\}} $ is not even.
We use $\gelo(E)$ to denote the set of all GELO partitions of $\Eindexset$, and if $p \in \N$, then $\gelo(p)$ is shorthand for $\gelo([p])$.
\end{definition}

\subsection{Assignments}

Our calculations of moments of demerit factors are based on enumeration of objects called assignments, which we now define.
\begin{definition}[Assignment]
Let $E$ be a subset of $\N$.
An {\it assignment for $E$} is an element of $\N^\Eindexset$, i.e., a function from $\Eindexset$ into $\N$.
If $\tau$ is an assignment for $E$ and $(e,s,v)\in\Eindexset$, then we usually write $\tau_{e,s,v}$ (or $\tau_{e s v}$) rather than the $\tau(e,s,v)$ to denote the value of $\tau$ at $(e,s,v)$.  We also introduce the following as compact notations for sets of assignments that will be useful in the rest of the paper:
\begin{itemize}
\item $\As(E)=\N^\Eindexset$, the set of all assignments for $E$,
\item $\As(E,=)=\{\tau\in\As(E): \tau_{e 0 0} + \tau_{e 0 1} = \tau_{e 1 0} + \tau_{e 1 1} \text{ for every } e \in E\}$,
\item $\As(E,\ell)=\{\tau \in \As(E): \tau(\Eindexset) \subseteq [\ell]\}$, and
\item $\As(E,=,\ell) = \As(E,=) \cap \As(E,\ell)$.
\end{itemize}
If $\cP \in \Part(E)$, then we also define:
\begin{itemize}
\item $\As(\cP)=\{\tau\in\As(E): \tau_{e s v}\equiv \tau_{f t w} \text{ if and only if } (e,s,v)\equiv(f,t,w)\pmod{\cP}\}$,
\item $\As(\cP,=)=\As(\cP)\cap\As(E,=)$,
\item $\As(\cP,\ell)=\As(\cP)\cap\As(E,\ell)$,
\item $\As(\cP,=,\ell)=\As(\cP)\cap\As(E,=)\cap\As(E,\ell)$,
\end{itemize}
where we note that we can deduce $E$ from $\cP$ since $E$ is the set of projections onto the first coordinate of all the triples in the all the classes of $\cP$.
Also note that $\As(\cP)$ is the set of all assignments for $E$ that induce the partition $\cP$.
\end{definition}
The last part of this definition associates certain assignments to partitions of $\Eindexset$.
This association is used to define an important class of partitions.
\begin{definition}[Satisfiable partition]
Let $E \subseteq \N$ and $\cP \in \Part(E)$.
We say that $\cP$ is {\it satisfiable} to mean that $\As(\cP,=)\not=\emptyset$.
The set of all satisfiable partitions of $\Eindexset$ is denoted $\Sat(E)$, and $\Sat(p)$ is a shorthand for $\Sat([p])$ when $p\in\N$.
\end{definition}
\begin{definition}[Contributory partition]
Let $p \in \N$ and $\cP \in \Part(p)$.
We say that $\cP$ is {\it contributory} to mean that it is both satisfiable and globally even, locally odd.
The set of all contributory partitions of $\indexset$ is denoted $\Con(p)$, so that $\Con(p)=\Sat(p)\cap\gelo(p)$.
\end{definition}
A precise formula for the $p$th central moment of the demerit factor was obtained in \cite[Prop.\ 3.1]{Katz-Ramirez} using the combinatorial objects that we have just defined.
\begin{proposition}\label{Sanri}
If $p,\ell\in\N$, then
\[
\mom{p} \ssac(f) = \sum_{\cP \in \Con(p)} |\As(\cP,=,\ell)|.
\]
\end{proposition}

\subsection{Isomorphism classes of partitions}

In many applications of \cref{Sanri}, we can find pairs of partitions $\cP,\cQ \in \Con(p)$ such that $|\As(\cP,=,\ell)|=|\As(\cQ,=,\ell)|$, and indeed there are often a large number of partitions that have the same number of assignments associated to them in the summation of \cref{Sanri}.
This is often due to underlying symmetries, the group of which we now describe.
If $\cP$ is a partition of $\indexset$ and $\tau$ is an assignment in $\As(\cP,=,\ell)$, then we can view $\tau \colon \indexset \to [\ell]$ as a means of assigning numbers to places in a system of $p$ equations:
\begin{align*}
\tau_{0,0,0} + \tau_{0,0,1} & = \tau_{0,1,0} + \tau_{0,1,1} \\
& \,\,\, \vdots \\
\tau_{p-1,0,0} + \tau_{p-1,0,1} & = \tau_{p-1,1,0} + \tau_{p-1,1,1},
\end{align*}
and the group of symmetries that we use is a collection of permutations of $\indexset$ that essentially do three things: (i) they can change the order in which the equations are listed, (ii) they can transpose the sides of any one of the equations, and (iii) they can transpose the two terms on a given side of any one of the equations.  In the next few paragraphs, we describe this group of symmetries formally.

If $X$ is a set, then $S_X$ denotes the group of all permutations of $X$.
If $p \in \N$, then $S_p$ is shorthand for $S_{[p]}$.
The group $S_{\indexset}$ acts on elements of $\indexset$ in the usual way: if $\pi \in S_{\indexset}$ and $(e,s,v) \in \indexset$, then the action maps $(e,s,v)$ to $\pi(e,s,v)$.  By extension $S_{\indexset}$ also acts on subsets of $\indexset$: if $P \subseteq \indexset$, then $\pi(P)=\{\pi(e,s,v): (e,s,v) \in P\}$.  By further extension, $S_{\indexset}$ acts on sets of subsets of $\indexset$: if $\cP$ is a set of subsets of $\indexset$, then $\pi(\cP)=\{\pi(P): P \in \cP\}$.  This extended action restricts to an action of $S_{\indexset}$ on $\Part(P)$ because application of a permutation $\pi$ from $S_{\indexset}$ will take a partition to a partition.  We even extend the action of $S_{\indexset}$ to sets of sets of subsets of $\indexset$: if $\fP$ is a set of sets of subsets of $\indexset$, then $\pi(\fP)=\{\pi(\cP): \cP \in \fP\}$.

If $\pi\in S_{\indexset}$, then we use the notation $\pi^*$ to denote the permutation $\tau\mapsto\tau\circ\pi$ of $\As([p])$; the inverse of $\pi^*$ is $(\pi^{-1})^*$.
For $\tau \in \As([p])$ and $(e,s,v) \in \indexset$, we have $(\pi^*(\tau))_{e,s,v}=(\tau\circ\pi)_{e,s,v} = \tau_{\pi(e,s,v)}$.

Our group of symmetries will be a subgroup of $S_{\indexset}$ defined using wreath products in stages.
For the first stage, we write the elements of the wreath product $S_2 \Wr_{[2]} S_2$ as pairs of the form $\delta=((\digamma_0,\digamma_1),\sigma)$, where $\digamma_0,\digamma_1,\sigma \in S_2$ and $\delta$ acts on $(s,v) \in \lindexset$ by the rule
\[
\delta(s,v)=\big((\digamma_0,\digamma_1),\sigma\big)(s,v)=\left(\sigma(s),\digamma_{\sigma(s)}(v)\right).
\]
This wreath product $S_2 \Wr_{[2]} S_2$ is isomorphic to the dihedral group of order $8$.
Now we consider a further wreath product: we define $\cWp$ to be the wreath product $(S_2 \Wr_{[2]} S_2) \Wr_{[p]} S_p$ whose elements are of the form $\pi=((\delta_0,\ldots,\delta_{p-1}),\epsilon)$ with $\delta_0,\ldots,\delta_{p-1} \in S_2 \Wr_{[2]} S_2$ and $\epsilon \in S_p$.  For each $e \in [p]$, we write $\delta_e = ((\digamma_{e,0},\digamma_{e,1}),\sigma_e)$, so that
\[
\pi = \bigg(\Big(\big((\digamma_{0,0},\digamma_{0,1}),\sigma_0\big),\ldots,\big((\digamma_{p-1,0},\digamma_{p-1,1}),\sigma_{p-1}\big)\Big),\epsilon\bigg).
\]
Then $\pi$ acts on each $(e,s,v) \in \lindexset$ by the rule
\begin{equation}\label{Mikhail}
\pi(e,s,v)=\left(\epsilon(e),\sigma_{\epsilon(e)}(s),\digamma_{\epsilon(e),\sigma_{\epsilon(e)}(s)}(v)\right).
\end{equation}
We say that {\it $\pi$ uses permutation $\epsilon$ to permute equations, then uses permutation $\sigma_e$ to permute the sides of equation $e$ for each $e \in [p]$, and then uses permutation $\digamma_{e,s}$ to permute the places on side $s$ of equation $e$ for each $e \in [p]$ and $s \in [2]$}.
Since each $\pi \in \cWp$ permutes $\indexset$, we identify it with a permutation in $S_{\indexset}$; this makes $\cWp$ a subgroup of $S_{\indexset}$.
We have
\begin{equation}\label{Samuel}
|\cWp|=p! 8^p
\end{equation}
by \cite[eq.\ (5)]{Katz-Ramirez}.

Since $\cWp$ is a subgroup of $S_{\indexset}$, it acts on elements of $\indexset$, subsets of $\indexset$, sets of subsets of $\indexset$, and sets of sets of subsets of $\indexset$ as described above.
From this action, we obtain an equivalence relation on $\Part(p)$ for each $p \in \N$, which we now define.
\begin{definition}[Isomorphic partitions]
Let $p \in \N$ and $\cP,\cQ \in \Part(p)$.
We say that $\cP$ and $\cQ$ are {\it isomorphic} and write $\cP \cong \cQ$ to mean that there is some $\pi \in \cWp$ such that $\cQ=\pi(\cP)$.
\end{definition}
One of the most important properties of the action of $\cWp$ on partitions is that it preserves global evenness, local oddness, satisfiability, number of assignments, and contributoriness, as shown in \cite[Lem.\ 4.8--Cor.\ 4.10]{Katz-Ramirez}, which we summarize here.
\begin{lemma} \label{canapes}
Let $p \in \N$. If $\cP, \cQ \in \Part(p)$ with $\cP \cong \cQ$.  Then we have the following:
\begin{enumerate}[label= (\roman*)]
\item $\cP \in \gelo(p)$ if and only if $\cQ \in \gelo(p)$;
\item for every $\ell \in \N$, we have $|\As(\cP,=,\ell)|=|\As(\cQ,=,\ell)|$, so that $\cP \in \Sat(p)$ if and only if $\cQ \in \Sat(p)$; and
\item $\cP \in \Con(p)$ if and only if $\cQ \in \Con(p)$.
\end{enumerate}
\end{lemma}
Therefore, we organize the contributory partitions into classes modulo this isomorphism relation, and notice that all partitions within an isomorphism class will have the same number of associated assignments.
We introduce notation that makes it easy to keep track of these.
\begin{definition}[$\Isom(p)$] Let $p \in \N$.  We use $\Isom(p)$ to denote the set of all equivalence classes within $\Con(p)$ under the isomorphism relation $\cong$.
\end{definition}
\begin{definition}[$\Sols(\cP,\ell)$]\label{Sonia} Let $p, \ell \in \N$.  If $\fP$ is any subset of $\Part(p)$ such that all the partitions in $\fP$ are isomorphic to each other, then we define $\Sols(\fP,\ell)$ to be the common value (by \cref{canapes}) of $|\As(\cP,=,\ell)|$ for $\cP \in \fP$.
\end{definition}
We most often apply \cref{Sonia} when $\fP \in \Isom(p)$, and this leads to our basic formula for the central moments of the sum of squares of autocorrelation in \cite[Prop.\ 4.13]{Katz-Ramirez}, which we restate here.
\begin{proposition}\label{Sanria}
If $ p,\ell \in \N$, then
\[ 
\mom{p} \ssac(f) = \sum_{ \fP \in \Isom(p)} |\fP| \Sols(\fP,\ell).
\]
\end{proposition}

\subsection{Structure of contributory partitions}

In addition to this basic formula, we shall use some constraints on the structure of contributory partitions that were proved in \cite{Katz-Ramirez}.
To state these constraints, we need another definition.
\begin{definition}[Twin class]
Let $p \in \N$.  A {\it twin class} is a subset $\cP$ of $\indexset$ such that there is some $(e,s) \in [p]\times [2]$ with $\{(e,s,0)$, $(e,s,1)\} \subseteq \cP$.
\end{definition}
The following constraint on the structure of contributory classes was proved as \cite[Lem.\ 5.2(i)]{Katz-Ramirez}.
\begin{lemma}\label{Epazote}
Let $p\in \N$ and $\cP \in \Con(p)$.
For each $e\in[p]$, the restriction $\cP_{\{e\}}$ is a partition of type $\ms{2,1,1}$ or $\ms{1,1,1,1}$.  In the former case, the class with cardinality $2$ is a twin class, i.e., there is some $s \in [2]$ such that $\cP_{\{e\}} = \{\{(e,s,0),(e,s,1)\}, \{(e,1-s,0)\},\{(e,1-s,1)\}\}$.  In the latter case, $\cP_{\{e\}}= \{\{(e,0,0)\},\{(e,0,1)\},\{(e,1,0)\},\{(e,1,1)\}\}$.
\end{lemma}

\subsection{Counting principle}

To help us count assignments, we also use an elementary counting result, which is proved in Lemma C.6 of \cite{Katz-Ramirez}.
\begin{lemma}\label{wowzers}
Let $\ell \in \N $. Then the number $(A,B,C,D) \in [\ell]^4$ with $A+B=C+D$ is $(2 \ell^3+\ell)/3$.
\end{lemma}

\section{Refinement of partitions}\label{Ulrich}
Now we introduce the notion of refinement of partitions, which is used to prove both that $\mom{p} \ssac(f)$ is a quasi-polynomial function of $\ell$ (see \cref{Quentin}) and that the standardized moments of $\ssac(f)$ tend to the same values as those of the standard normal distribution as the length of binary sequences tends to infinity (see \cref{Andrew}).
First, we introduce the following notion.
\begin{definition}[Coarser and finer partitions]
Let $\cP$ and $\cQ$ be partitions of the same set $S$.
If each class of $\cP$ is a subset of a class of $\cQ$, we say that $\cP$ is {\it finer than} (or {\it a refinement of}) $\cQ$ and write $\cP \preceq \cQ$, or equivalently, that $\cQ$ is {\it coarser than} (or \it{a coarsening of}) $\cP$ and write $\cQ \succeq \cP$.
If, in addition, $\cP\not=\cQ$, then we say that $\cP$ is {\it strictly finer than} (or {\it a strict refinement of}) $\cQ$ and write $\cP \prec \cQ$, or equivalently, that $\cQ$ is {\it strictly coarser than} (or \it{a strict coarsening of}) $\cP$ and write $\cQ \succ \cP$.
\end{definition}
We note that the finer than relation is preserved under the action of $\cWp$.
\begin{lemma}\label{Billy}
Let $\cQ,\cP \in \Part(p)$ and $\pi \in \cWp$. We have $\cP = \cQ$ if and only if $\pi(\cP) = \pi(\cQ)$.
Furthermore,  $\cP \preceq \cQ$ if and only if $\pi(\cP) \preceq \pi(\cQ)$.
Therefore, we have $\cP \prec \cQ$ if and only if $\pi(\cP) \prec \pi(\cQ)$.
\end{lemma}
\begin{proof}
The first statement is true because $\pi$ acts as a permutation on the set $\Part(p)$. 

Now suppose that $\cP \preceq \cQ$, and let $S$ be a class in $\pi(\cP)$.
Then $S=\pi(P)$ for some $P \in \cP$.
Then there is some $Q \in \cQ$ such that $P \subseteq Q$.
So then $S=\pi(P) \subseteq \pi(Q) \in \pi(\cQ)$.
Thus, $\cP \preceq \cQ$ implies $\pi(\cP) \preceq \pi(\cQ)$.

If $\pi(\cP) \preceq \pi(\cQ)$, then since $\pi^{-1} \in \cWp$, the conclusion of the previous paragraph shows that $\pi^{-1}(\pi(\cP)) \preceq \pi^{-1}(\pi(\cQ))$, i.e., $\cP \preceq \cQ$.
\end{proof}
Given a positive integer $p$, a partition $\cP \in \Part(p)$, and an $\ell \in \N$, our formulae for central moments require us to enumerate $\tau \in \As(\cP,=,\ell)$, i.e., $\tau\in\As(=,\ell)$ such that $\tau_\beta=\tau_\gamma$ if and only if $\beta\equiv\gamma\pmod{\cP}$.
We have seen in \cite{Katz-Ramirez} that this can be difficult even with a $p$ as small as $3$.
It turns out to be much easier to enumerate $\tau\in\As(=,\ell)$ such that $\tau_\beta=\tau_\gamma$ if $\beta\equiv\gamma\pmod{\cP}$ (note that the condition here is not biconditional).
This prompts the following notations.
\begin{notation}\label{Gerald}
Let $p, \ell \in \N$ and $E \subseteq[p]$, and $\cP \in \Part(E)$.
Then we use the following notations:
\begin{itemize}
\item $\As(\succeq\!\cP)=\bigcup_{\cQ \succeq \cP} \As(\cQ) = \{\tau \in \As(E): \tau_\beta = \tau_\gamma \text{ if }  \beta \equiv \gamma \!\! \pmod\cP\}$,
\item $\As(\succeq\!\cP,\ell)=\bigcup_{\cQ \succeq \cP} \As(\cQ,\ell) = \As(\succeq\!\cP) \cap \As(E,\ell)$, 
\item $\As(\succeq\!\cP,=)=\bigcup_{\cQ \succeq \cP} \As(\cQ,=)= \As(\succeq\!\cP) \cap \As(E,=)$, and
\item $\As(\succeq\!\cP,=,\ell)=\bigcup_{\cQ \succeq \cP} \As(\cQ,=,\ell)= \As(\succeq\!\cP) \cap \As(E,=,\ell)$,
\end{itemize}
and we note that each of these four unions is a union of disjoint sets since each assignment induces a unique partition.
\end{notation}
We record without proof some immediate consequences of \cref{Gerald}.
\begin{lemma} \label{Gertrude}
Let $\cP\in\Part(p)$.
\begin{enumerate}[label= (\roman*)]
\item\label{Solomon} Then \[
|\As(\succeq\!\cP,=,\ell)| = \sum_{\cQ \succeq \cP} |\As(\cQ,=,\ell)|.
\]
\item Therefore,
\[
|\As(\succeq\!\cP,=,\ell)| \geq |\As(\cP,=,\ell)|.
\]
\end{enumerate}
\end{lemma}
\begin{lemma}\label{Julian}
If $\cP, \cQ \in\Part(p)$ with $\cP \preceq \cQ$, then
\[
\As(\succeq\!\cP,=,\ell) \supseteq \As(\succeq\!\cQ,=,\ell)
\]
and
\[
|\As(\succeq\!\cP,=,\ell)| \geq |\As(\succeq\!\cQ,=,\ell)|.
\]
\end{lemma}
We also note that action of $\cWp$ on partitions preserves the sizes of the newly defined sets of assignments.
\begin{lemma}\label{Esmeralda}
Let $p,\ell \in \N$, $\cP \in \Part(p)$, and $\pi \in \cWp$.
Then
\[
|\As(\succeq\!\cP,=,\ell)|=|\As(\succeq\!\pi(\cP),=,\ell)|.
\]
\end{lemma}
\begin{proof}
We have
\begin{align*}
|\As(\succeq\!\pi(\cP),=,\ell)| 
&= \sums{\cR \in \Part(p) \\ \cR \succeq \pi(\cP)} |\As(\cR,=,\ell)|&&\text{ by \cref{Gertrude}}  \\
&= \sums{\cQ \in \Part(p) \\ \pi(\cQ) \succeq \pi(\cP)} |\As(\pi(\cQ),=,\ell)| && \text{ because $\pi$ acts on $\Part(p)$} \\
&= \sums{\cQ \in \Part(p) \\ \cQ \succeq \cP} |\As(\pi(\cQ),=,\ell)| && \text{ by \cref{Billy}} \\
&= \sums{\cQ \in \Part(p) \\ \cQ \succeq \cP} |\As(\cQ,=,\ell)| &&\text{ by \cref{canapes}} \\
&= |\As(\succeq\!\cP,=,\ell)| && \text{ by \cref{Gertrude}}. \qedhere
\end{align*}
\end{proof}

\section{Central Moments are Quasi-Polynomial}\label{Quentin}

In this section, we prove \cref{Florence}, which states that $\mom{p} \ssac{f}$ is a quasi-polynomial function of $\ell$ in which all the polynomial coefficients are rational numbers.
We begin by showing that certain critical sets of assignments have cardinalities of this form.
\begin{lemma}\label{Rudolph}
Let $p \in \N$ and let $\cP \in \Part(p)$.
Then $|\As(\succeq\!\cP,=,\ell)|$ is a quasi-polynomial function of $\ell$ in which all coefficients are rational.
\end{lemma}
\begin{proof}
The set $\As(\succeq\!\cP,=,\ell)$ is the set of $\tau \in \N^{\indexset}$ meeting the following conditions:
\begin{itemize}
\item $\tau_\beta = \tau_\gamma$ if $\beta \equiv \gamma \pmod\cP$,
\item $\tau_{e,0,0}+\tau_{e,0,1}=\tau_{e,1,0}+\tau_{e,1,1}$ for every $e \in E$, and
\item $0 \leq \tau_{e,s,v} \leq \ell-1$ for every $(e,s,v) \in \indexset$.
\end{itemize}
If $p$ is positive, then $|\As(\succeq\!\cP,=,\ell)|$ is counting the number of integer points in the intersection of some hyperplanes defined by the equations in the first two points and the closed $4 p$-dimensional hypercube with side length $\ell-1$ described in the third point. 
Thus, Ehrhart theory \cite[Thm.\ 3.23]{Beck-Robins} tells us that $|\As(\succeq\!\cP,=,\ell)|$ is a quasi-polynomial function of $\ell$, and the coefficients in the quasi-polynomial must all be rational because the quasi-polynomial maps $\N$ into $\N$.  If $p=0$, then for every $\ell \in \N$ the set $\As(\succeq\!\cP,=,\ell)$ contains only the empty function from $\emptyset$ into $\N$, so its cardinality is the constant quasi-polynomial $1$.
\end{proof}
We can now use M\"obius inversion on the poset $\Part(p)$ (equipped with the reverse ordering relation to $\preceq$) to see that sets of the form $\As(\cP,=,\ell)$ have cardinalities that are quasi-polynomial in $\ell$.
\begin{lemma}\label{Daphne}
Let $p \in \N$ and let $\cP \in \Part(p)$.
Then $|\As(\cP,=,\ell)|$ is a quasi-polynomial function of $\ell$ in which all coefficients are rational.
\end{lemma}
\begin{proof}
First we recall from \cref{Gertrude}\ref{Solomon} that 
\[
|\As(\succeq\!\cQ,=,\ell)| = \sum_{\cR \succeq \cQ} |\As(\cR,=,\ell)|
\]
for every $\cQ \in \Part(p)$.
We define $\sqsubseteq$ to be $\succeq$ (and so $\sqsupseteq$ is $\preceq$).  So $(\Part(p),\sqsubseteq)$ is a poset with minimum element $\cM=\{\indexset\}$.  So then
\begin{equation}\label{lettuce}
|\As(\succeq\!\cQ,=,\ell)| = \sum_{\cR \sqsubseteq \cQ} |\As(\cR,=,\ell)|
\end{equation}
for every $\cQ \in \Part(p)$.
Now we set $g\colon \Part(P)\times\Part(P) \to \N$ with
\[
g(\cP,\cQ) = \begin{cases}
|\As(\cQ,=,\ell)| & \text{if $\cP=\cM$,} \\
0 & \text{if $\cP\not=\cM$.}
\end{cases}
\]
We also set $f\colon \Part(P)\times\Part(P) \to \N$ with
\[
f(\cP,\cQ) = \begin{cases}
|\As(\succeq\!\cQ,=,\ell)| & \text{if $\cP=\cM$,} \\
0 & \text{if $\cP\not=\cM$.}
\end{cases}
\]
So now, equation \eqref{lettuce} becomes
\[
f(\cM,\cQ) = \sum_{\cM \sqsubseteq \cR \sqsubseteq \cQ} g(\cM,\cR)
\]
for all $\cQ \in \Part(p)$.
Now consider the equation
\[
f(\cP,\cQ) = \sum_{\cP \sqsubseteq \cR \sqsubseteq \cQ} g(\cP,\cR),
\]
which we know is true when $\cP=\cM$ from the previous equation, and which is obviously true when $\cP\not=\cM$ because in that case the definitions of $f$ and $g$ make both sides obviously zero.
By using M\"obius inversion (see \cite[Prop.\ 12.7.3]{Cameron}), we have, for every $\cP, \cQ \in \Part(p)$, 
\[
g(\cP,\cQ) = \sum_{\cP \sqsubseteq \cR \sqsubseteq \cQ} f(\cP,\cR) \mu(\cR,\cQ),
\]
where $\mu \colon \Part(p) \times \Part(p) \to \Z$ is the M\"obius $\mu$ function for the incidence algebra for the poset $(\Part(p),\sqsubseteq)$.
Therefore, for every $\cQ \in \Part(p)$, we have
\[
g(\cM,\cQ) = \sum_{\cM \sqsubseteq \cR \sqsubseteq \cQ} f(\cM,\cR) \mu(\cR,\cQ),
\]
which means
\[
|\As(\cQ,=,\ell)| = \sum_{\cR \succeq \cQ} |\As(\succeq\!\cR,=,\ell)| \mu(\cR,\cQ),
\]
so that $|\As(\cQ,=,\ell)|$ is a $\Z$-linear combination of values of the form $|\As(\succeq\!\cR,=,\ell)|$ for $\cR \in \Part(p)$.  By \cref{Rudolph} these values are quasi-polynomial functions of $\ell$ with rational coefficients.
Therefore, for every $\cQ \in \Part(p)$, we know that $|\As(\cQ,=,\ell)|$ is a quasi-polynomial function of $\ell$ with rational coefficients.
\end{proof}
Using this last lemma, we can prove \cref{Florence}.
\begin{theorem}
Let $p \in \N$.  Then $\mom{p} \ssac(f)$ is a quasi-polynomial function of $\ell$ whose coefficients are all rational numbers.
Therefore, $\mom{p} \ADF(f)$ is $\ell^{-2 p}$ times a quasi-polynomial function of $\ell$ whose coefficients are all rational numbers.
\end{theorem}
\begin{proof}
The formula for $\mom{p} \ssac(f)$ given in \cref{Sanri} shows that it is a finite sum of quantities of the form $|\As(\cP,=,\ell)|$ for partitions $\cP \in \Part(p)$, all of which are known to be quasi-polynomial functions of $\ell$ with rational coefficients by \cref{Daphne}.
The statement about the moments of $\ADF$ then follows by applying formula \eqref{Natasha} from the Introduction.
\end{proof}

\section{Asymptotics}\label{Andrew}

This section is dedicated to the proof of \cref{Andy}.
In \cref{Samantha} we define special kinds of partitions, including a kind (called {\it principal partitions}) that produce the asymptotically dominant contribution to the moments, as is shown in \cref{Conrad}.
We conclude with the determination of the asymptotic behavior of the moments in \cref{Abraham}.

\subsection{Separable and principal partitions}\label{Samantha}

In this section, we define special kinds of partitions that help us prove our asymptotic estimates of the standardized moments of the demerit factor.
Some partitions, which we will call {\it separable partitions}, contribute to the moments in ways that are easier to quantify, and among these are the {\it principal partitions}, whose contributions dominate those of all other contributory partitions.
Before we can define these two kinds of partitions, we need some preliminary definitions to describe different kinds of partitions.
Throughout this section, whenever ${\mathcal X}$ is a set of sets, $\cup{\mathcal X}$ is shorthand for $\bigcup_{X \in {\mathcal X}} X$.
\begin{definition}[Equation-disjoint]
Let $p \in \N$, $\cP \in \Part(p)$, and $P,Q \in \cP$.
We say that $P$ and $Q$ are \textit{equation-disjoint} to mean that for every $e \in [p]$, at least one of $(\eindexset) \cap P$ or $(\eindexset) \cap Q$ is empty.
\end{definition}
\begin{definition}[Equation-identical]
Let $p \in \N$, $\cP \in \Part(p)$, and $P,Q \in \cP$.
We say that $P$ and $Q$ are \textit{equation-identical} to mean that for every $e \in [p]$, the intersection $(\eindexset) \cap P$ is empty if and only if $(\eindexset) \cap Q$ is empty.
\end{definition}
\begin{definition}[Equation-imbricate]
Let $p \in \N$, $\cP \in \Part(p)$, and $P,Q \in \cP$.
We say that $P$ and $Q$ are \textit{equation-imbricate} to mean that they are neither equation-disjoint or equation-identical.
\end{definition}
We record without proof the evident fact that equation-identicality is an equivalence relation, and introduce the quotient by this relation.
\begin{lemma}
Let $p \in\N$ and $\cP \in \Part(p)$.  Then the equation-identicality relation is an equivalence relation on $\cP$.
\end{lemma}
\begin{definition}[Equation-identicality class and quotient partition]
Let $p \in \N$ and $\cP \in \Part(p)$.
An equivalence class of the equation-identical relation on $\cP$ is called an {\it equation-identicality class} of $\cP$.
The partition of $\cP$ into equation-identicality classes is called the {\it quotient partition} (or just {\it quotient}) of $\cP$.
\end{definition}
Now we are ready to introduce the concepts of separable and principal partitions.
\begin{definition}[Separable partition]
If $p \in \N$ and $\cP \in \Part(p)$, then we say that $\cP$ is {\it separable} to mean that no two classes of $\cP$ are equation-imbricate.
We use $\Sep(p)$ to denote the set of all separable partitions in $\Part(p)$.
\end{definition}
\begin{definition}[Principal partition]
If $p\in\N$ and $\cP \in \Con(p)\cap\Sep(p)$, then we say that $\cP$ is {\it principal} to mean that $|P|=2$ for every $P \in \cP$.
\end{definition}	
Now we recognize that the action of $\cWp$ on partitions preserves all the concepts that we have introduced above.
\begin{lemma}\label{Justinian}
Let $p\in\N$, $\cP\in\Part(p)$, $P,Q\in \cP$, and $\pi\in\cWp$.
Then $P$ and $Q$ are equation-disjoint (resp., equation-identical, equation-imbricate) if and only if $\pi(P)$ and $\pi(Q)$ are equation-disjoint (resp., equation-identical, equation-imbricate).
\end{lemma}
\begin{proof}
Let $\pi=\left((\delta_0,\ldots,\delta_{p-1}),\epsilon\right)$.
We note that the action \eqref{Mikhail} of $\pi$ implies that $\pi(\eindexset) \subseteq \eeindexset$, but $\pi$ is a permutation and $|\eindexset|=|\eeindexset|=4$, so that
\begin{equation}\label{Nikita}
\pi(\eindexset) = \eeindexset.
\end{equation}

Suppose that $P$ and $Q$ are equation-disjoint.
This is equivalent to saying that for every $e \in  [p] $ either $(\eindexset) \cap P$ is empty or $(\eindexset)\cap Q$ is empty.
Because $\pi$ is a bijection, this is true if and only if for every $e \in [p]$ either $\pi(\eindexset) \cap \pi(P)$ is empty or $\pi(\eindexset) \cap \pi(Q)$ is empty.
By \eqref{Nikita} this is true if and only if for every $e \in  [p] $ either $(\eeindexset) \cap \pi(P)$ is empty or $(\eeindexset) \cap \pi(Q)$ is empty.
Since $\epsilon$ is a bijection, this is true if and only if for every $e \in [p] $ either $(\eindexset) \cap \pi(P)$ is empty or $(\eindexset) \cap \pi(Q)$ is empty.
This is equivalent to saying that $\pi(P)$ and $\pi(Q)$ are equation-disjoint.
Hence $P$ and $Q$ are equation-disjoint if and only if $\pi(P)$ and $\pi(Q)$ are equation-disjoint.
	
Suppose $P$ and $Q$ are equation-identical.
This is equivalent to saying that for every $e \in [p]$ we have $(\eindexset) \cap P$ is empty if and only if $(\eindexset) \cap Q$ is empty.
Because $\pi$ is a bijection, this is true if and only if for every $e \in [p]$ we have $\pi(\eindexset) \cap \pi(P)$ is empty if and only if $\pi(\eindexset) \cap \pi(Q)$ is empty.
By \eqref{Nikita}, this is true if and only if for every $e \in [p]$ we have $(\eeindexset)\cap \pi(P)$ is empty if and only if $(\eeindexset)\cap \pi(Q)$ is empty.
Since $\epsilon$ is a bijection, this is true if and only if for every $e \in [p]$ we have $(\eindexset) \cap \pi(P)$ is empty if and only if $(\eindexset)\cap \pi(Q)$ is empty.
This is equivalent to saying that $\pi(P)$ and $\pi(Q)$ are equation-identical.
Hence $P$ and $Q$ are equation-identical if and only if $\pi(P)$ and $\pi(Q)$ are equation-identical.

Since equation-imbricate classes are precisely those classes that are neither equation-disjoint nor equation-identical, the claim that $P$ and $Q$ are equation imbricate if and only if $\pi(P)$ and $\pi(Q)$ are equation-imbricate follows from the previous two claims.
\end{proof}
\begin{corollary}\label{Lemon}
Let $p \in \N$, $\cP \in \Part(p)$, and $\pi\in \cWp$.
Then $\cP \in \Sep(p)$ if and only if $\pi(\cP)\in\Sep(p)$.
Furthermore, $\cP$ is principal if and only if $\pi(\cP)$ is principal.
\end{corollary}
\begin{proof}
\cref{Justinian} implies that $\pi$ preserves separability of partitions.
From this and the fact that $\pi$ preserves contributoriness (see \cref{canapes}) and sizes of classes in partitions ($\pi$ is a bijection), we also see that $\pi$ preserves principality.
\end{proof}
Another straightforward consequence of \cref{Justinian} is that the action of $\cWp$ is compatible with quotients of partitions.
\begin{corollary}
Let $p \in \N$, $\cP \in \Part(p)$, and $\mathfrak{P}$ be the quotient of $\cP$.
For any $\pi \in \cWp$, the quotient of $\pi(\cP)$ is $\pi(\mathfrak{P})$.
\end{corollary}
The property of separability imposes a great deal of structure upon a contributory partition.
\begin{lemma} \label{Rosemary} 
Let $p \in \N$, let $\cP\in\Con(p)\cap\Sep(p)$, and let $\mathfrak{P}$ be the quotient of $\cP$.
Then each class of $\mathfrak{P}$ is of size $3$ or $4$.
\begin{enumerate}[label= (\roman*)]
\item If $\cC \in \mathfrak{P}$ with $|\cC| = 3$, then there is some nonempty $E \subseteq [p]$ of even cardinality such that one set in $\cC$ (which we call $C_0$) is of size $2 |E|$ and the other two sets in $\cC$ (which we call $C_{1,0}$ and $C_{1,1}$) are each of size $|E|$.
Furthermore, for each $e \in E$, there is some $s_e \in [2]$ such that $(e,s_e,0), (e,s_e,1) \in C_0$ and one of $(e,1-s_e,0)$ and $(e,1-s_e,1)$ lies in $C_{1,0}$ while the other lies in $C_{1,1}$.
\item\label{parsley} If $\cC \in \mathfrak{P}$ with $|\cC| = 4$, then there is some nonempty $E \subseteq [p]$ of even cardinality such that all four sets in $\cC$ are of size $|E|$.
Furthermore, there is some labeling $C_{0,0}$, $C_{0,1}$, $C_{1,0}$, and $C_{1,1}$ of the four classes in $\cC$ such that for every $e \in E$, there is some $s_e \in [2]$ such that one of $(e,s_e,0)$ or $(e,s_e,1)$ lies in $C_{0,0}$ while the other lies in $C_{0,1}$, and one of $(e,1-s_e,0)$ or $(e,1-s_e,1)$ is in $C_{1,0}$ while the other lies in $C_{1,1}$.
\end{enumerate}
\end{lemma}
\begin{proof}
Suppose that $\cC \in \fP$.
Let $E$ be the set of all $e \in [p]$ such that $\eindexset$ has a nonempty intersection with one (and therefore all) of the sets in $\cC$.
Of course, $E$ must be nonempty since the sets in $\cC$ are nonempty (since they are classes in the partiton $\cP$).
Since $\cP \in \Sep(p)$, this means that $E\times\lindexset$ is disjoint from all the sets in $\cP\smallsetminus\cC$.
Thus, for each $e \in E$, we have $\cP_{\{e\}}=\cC_{\{e\}}$, which is a partition of the set $\eindexset$.
Since $\cP \in \Con(p)$, then by \cref{Epazote}, for each $e \in [p]$ we know that $\cP_{\{e\}}=\cC_{\{e\}}$ is either type $\ms{2,1,1}$ or $\ms{1,1,1,1}$.
This implies that $|\cC|$ is either $3$ or $4$.
		
Let us consider the case where $|\cC| = 3$.
Suppose for contradiction that $\cC$ contained two distinct classes, say $P$ and $Q$, such that $|P_{\{f\}}|=|Q_{\{g\}}|=2$ for some $f,g \in [p]$.
Let $R$ be the third class in $\cC$.
Note that $f\not=g$ since $\cP_{\{f\}}$ and $\cP_{\{g\}}$ must be type $\ms{2,1,1}$ by \cref{Epazote}.
Furthermore, \cref{Epazote} shows that there are some $s,t \in [2]$ such that $(f,s,0),(f,s,1) \in P$ and $(g,t,0),(g,t,1) \in Q$.
Since $P$, $Q$, and $R$ are equation-identical, this means that $|P_{\{g\}}|=|Q_{\{f\}}|=|R_{\{f\}}|=|R_{\{g\}}|=1$.
Since $\cP \in \Con(p) \subseteq \Sat(p)$, there is some $\tau \in \As(\cP,=)$.
Let $A$ (resp., $B$, $C$) denote the element of $\N$ such that $\tau_\gamma=A$ (resp., $B$, $C$) for $\gamma\in P$ (resp., $Q$, $R$).
Then looking at $\gamma$'s of the form $(f,*,*)$ (resp., $(g,*,*)$), we must have $2 A=B+C$ and $2 B=A+C$, which would make $A=B$, which contradicts $\cP \in \Sat(p)$.
So we conclude that if $|\cC|=3$, then $\cC$ has at most one class $C$ such that $|C_{\{e\}}|=2$ for some $e \in E$.
Recall that for each $e \in E$, the restriction $\cC_{\{e\}}$ is a partition of the set $\eindexset$ into at most $|\cC|=3$ subsets.
So for each $e \in E$, two elements of $\eindexset$ must lie in one of the classes of $\cC$.
Since there can be at most one such class in $\cC$, this shows that there is some class $C_0 \in \cC$ such that for every $e \in E$, the class $C_0$ must contain two elements of $\eindexset$ (making $|C_0|=2|E|$), and furthermore, by \cref{Epazote}, we know that for each $e \in E$, there is some $s_e \in [2]$ such that $(e,s_e,0)$ and $(e,s_e,1)$ both lie in $C_0$.
We name the other two sets in $\cC$ as $C_{1,0}$ and $C_{1,1}$, and then for each $e \in E$, one of $(e,1-s_e,0)$ and $(e,1-s_e,1)$ lies in $C_{1,0}$ and the other lies in $C_{1,1}$.
Then $|C_{1,0}|=|C_{1,1}|=|E|$, so that $|E|$ is even because $C_{1,0} \in \cP \in \Con(p) \subseteq \gelo(p)$.

Let us consider the case where $|\cC| = 4$.
Since $E$ is nonempty, there is some $f \in E$, and we let $C_{0,0}$ (resp., $C_{0,1}$, $C_{1,0}$, $C_{1,1}$) be the classes in $\cC$ such that $(f,0,0) \in C_{0,0}$ (resp., $(f,0,1) \in C_{0,1}$, $(f,1,0) \in C_{1,0}$, $(f,1,1) \in C_{1,1}$).
Suppose for contradiction that there is some $g \in E $, and $t \in [2] $ such that $|(\{g\} \times \{t\} \times [2]) \cap (C_{0,0} \cup C_{0,1})|=1$.
Since $\cP \in \Con(p) \subseteq \Sat(p)$, there is some $\tau \in \As(\cP,=)$.
Let $A$ (resp., $B$, $C$, $D$) denote the element of $\N$ such that $\tau_\gamma=A$ (resp., $B$, $C$, $D$) for $\gamma\in C_{0,0}$ (resp., $C_{0,1}$, $C_{1,0}$, $C_{1,1}$).
Then looking at $\gamma$'s of the form $(f,*,*)$ (resp., $(g,*,*)$), we must have $A+B=C+D$ (resp., either $A+C=B+D$ or $A+D=B+C$), which would make either $A=D$ or $A=C$, which contradicts $\cP \in \Sat(p)$.
Thus, for every $e \in E$ and $s \in [2]$, we must have $|(\{e\}\times\{s\}\times[2]) \cap (C_{0,0}\cup C_{0,1})|=0$ or $2$.
So for every $e \in E$, there is some $s_e \in [2]$ such that $|(\{e\}\times\{s_e\}\times[2]) \cap (C_{0,0}\cup C_{0,1})|=2$, so that one of $(e,s_e,0)$ or $(e,s_e,1)$ lies $C_{0,0}$ and the other lies in $C_{0,1}$, and one of $(e,1-s_e,0)$ or $(e,1-s_e,1)$ lies in $C_{1,0}$ and the other lies in $C_{1,1}$.
Then $|C_{0,0}|=|C_{0,1}|=|C_{1,0}|=|C_{1,1}|=|E|$, so that $|E|$ is even because $C_{0,0} \in \cP \in \Con(p) \subseteq \gelo(p)$.
\end{proof}
There are many conditions equivalent to principality; we list some important ones here.
\begin{lemma}\label{water}
Let $\cP \in \Con(p)\cap\Sep(p)$ and let $\fP$ be the quotient of $\cP$.  Then the following are equivalent:
\begin{enumerate}[label={(\roman*)}]
\item\label{rutabaga} $\cP$ is principal.
\item\label{banana} Every class in $\cP$ has cardinality $2$ and is not a twin class.
\item\label{broccoli} Every class in $\cP$ has cardinality $2$ and $|\cC|=4$ for every $\cC \in \fP$.
\item\label{zucchini} $|\cP|=2 p$ and $|\fP|=p/2$.
\item\label{tomato} $|\cP|-|\fP|=3p/2$.
\item\label{potato} $|\cP|-|\fP| \geq 3p/2$.
\item\label{turnip} $|\cP|=2 p$.
\end{enumerate}
\end{lemma}
\begin{proof}
Suppose that $\cP$ is principal.
Then every class in $\cP$ has cardinality $2$.
If $\cP$ had a twin class, then by \cref{Rosemary} there must be some $\cC \in \fP$ with $|\cC|=3$ and some $C_0 \in \cC \subseteq \cP$ with $|C_0|$ equal to twice a positive even integer, i.e., $|C_0| \geq 4$, which contradicts what we know about the sizes of classes of $\cP$.
So \ref{rutabaga} implies \ref {banana}.

Suppose that every class in $\cP$ has cardinality $2$ and is not a twin class.
Since $\cP \in \Con(p)\cap\Sep(p)$ and has no twin classes, then by \cref{Rosemary} we know that $|\cC|=4$ for each $\cC \in \fP$.
Hence, \ref{banana} implies \ref{broccoli}.

Suppose that every class in $\cP$ has cardinality $2$ and every class in the quotient $\fP$ of $\cP$ has cardinality $4$.
Since $\cP$ is a partition of $\indexset$, which has $4 p$ elements, into classes of size $2$, this makes $|\cP|=4 p/2=2 p$, and then $\fP$ is a partition of $\cP$ into classes of size $4$, so $|\fP|=2p/4=p/2$.
Hence, \ref{broccoli} implies \ref{zucchini}. 

Now, \ref{zucchini} directly implies \ref{tomato}.
Then \ref{tomato} clearly implies \ref{potato}.

Suppose $|\cP|-|\fP| \geq 3p/2$. Since $ \fP $ is a partition of $ \cP $, then we know that by \cref{Rosemary}
\[ \frac{|\cP| }{4} \leq \frac{|\cP| }{\max_{C \in \fP} |C|} \leq |\fP| . \]
Since $ \cP $ is a partition of $ \indexset $ with all classes of positive even size, then 
\[ \frac{3p}{2} \geq \frac{3}{4} \cdot  \frac{|\indexset|}{\min_{P \in \cP} |P|} \geq \frac{3|\cP|}{4} = |\cP| -  \frac{|\cP| }{4} \geq |\cP| - |\fP| \geq \frac{3p}{2}, \]
so all the inequalities are in fact equalities, and then $3|\cP|/4=3p/2$, so we conclude that $ |\cP| = 2p $. Thus, \ref{potato} implies \ref{turnip}.

Suppose $|\cP| =2p$.
Since each class must be of positive even size and $\cP$ is a partition of a set with cardinality equal to $4 p$, then we can conclude that every class in $\cP$ has size two.
Hence, \ref{turnip} implies \ref{rutabaga}.
\end{proof}
It now becomes clear that there are no principal partitions involved when computing odd moments.
\begin{corollary}\label{milk}
If $\cP \in \Con(p)\cap\Sep(p)$ and $\cP$ is not principal then
\[ |\cP|-|\fP| \leq  \ceil{ 3p/2} -1. \]
Moreover, if $p$ is odd then $\Part(p)$ has no principal partitions.
\end{corollary}
\begin{proof}
The first statement follows immediately from \cref{water}\ref{potato}.
If $ \cP \in \Part(p)$ is principal, then $\cP \in \Con(p)\cap\Sep(p)$ by definition.
So if $ \fP $ is the quotient of $ \cP $, then by \cref{water}\ref{zucchini} $ |\fP| = p/2 $.
Since $|\fP|$ must be a whole number, if $ p $ is odd then $\Part(p)$ has no principal partitions.
\end{proof}
When computing an even moment, \cref{Horatio} says that the principal partitions lie in a single isomorphism class, which we enumerate in \cref{Henrietta}.
\begin{lemma}\label{Horatio}
Let $p$ be a positive even integer and for each $(e,s,v) \in [p/2] \times \lindexset$, let $Q_{e,s,v} = \{ (e,s,v) , ( e + p/2, s,v) \}$ and let $\cQ=\{Q_{e,s,v} : (e,s,v) \in [p/2] \times \lindexset\}$.
Then a partition in $\Part(p)$ is principal if and only if it is isomorphic to $\cQ$.
\end{lemma}
\begin{proof}
We know that $\cQ\in\gelo(p)$ because every class in $\cQ$ is of size $2$ and $\card{(Q_{e,0,0})_{\{e\}}} = \card{\{(e,0,0)\}}=1$ and $\card{(Q_{e,0,0})_{e+p/2}}=\card{\{(e+p/2,0,0)\}}=1$ for every $e \in [p/2]$.
If we let $\tau\in\As([p])$ where $\tau_{e,0,0}=\tau_{e+p/2,0,0}=4 e$, $\tau_{e,0,1}=t_{e+p/2,0,1}=4 e+3$, $\tau_{e,1,0}=\tau_{e+p/2,1,0}=4 e+1$, and $\tau_{e,1,1}=\tau_{e+p/2,1,1}=4 e+2$ for every $e \in [p/2]$, then one sees that $\tau\in\As(\cQ,=)$, so that $\cQ\in\Sat(p)$, and so $\cQ\in\Con(p)$.
For $(e,s,v),(f,t,w) \in [p/2]\times\lindexset$, the sets $Q_{e,s,v}$ and $Q_{f,t,w}$ are equation-identical if $e=f$ and equation-disjoint otherwise.
So $\cQ$ is separable and therefore principal since every class in $\cQ$ is of size $2$.
Therefore, every partition isomorphic to $\cQ$ is also principal by \cref{Lemon}.

Now assume that $\cP$ is a principal partition in $\Part(p)$, and we shall design a permutation $\pi \in \cWp$ such that $\pi(\cQ)=\cP$.
Let $\fP$ be the quotient of $\cP$.
Then by \cref{water}\ref{broccoli} and \ref{zucchini}, $|\cP|=2 p$ with every class in $\cP$ of cardinality $2$ and $|\fP|=p/2$ with every class in $\fP$ of cardinality $4$.
Let us label the classes in $\fP$ as $\cC_0, \dots, \cC_{p/2 -1}$.
By \cref{Rosemary}\ref{parsley}, for each $i \in [p/2]$ there is some $E_i \subseteq [p]$ of even cardinality such that we can label the four elements of $\cC_i$ as $C_{i,0,0}$, $C_{i,0,1}$, $C_{i,1,0}$, and $C_{i,1,1}$, and there will be some $s_e,v_{e,0},v_{e,1} \in [2]$ for each $e \in E_i$ such that
\begin{align*}
(e,s_e,v_{e,0}) & \in C_{i,0,0} \\
(e,s_e,1-v_{e,0}) & \in C_{i,0,1} \\
(e,1-s_e,v_{e,1}) & \in C_{i,1,0} \\
(e,1-s_e,1-v_{e,1}) & \in C_{i,1,1}
\end{align*}
for each $e \in E_i$.
Since each of the four classes of the form $C_{i,*,*}$ is of cardinality $2$, this means $|E_i|=2$.
For each $i \in [p/2]$, write the two distinct elements of $E_i$ as $e_{i,0}$ and $e_{i,1}$.
We now choose some $\epsilon \in S_{[p]}$ such that for each $(i,j) \in [p/2]\times[2]$ we have $\epsilon(i+j(p/2))=e_{i,j}$.
For each $e \in [p]$, let $\sigma_e \in S_{[2]}$ be chosen with $\sigma_e(0)=s_e$.
For each $e \in [p]$ and $s \in [2]$, let $\digamma_{e,s} \in S_{[2]}$ be chosen with $\digamma_{e,s}(0)=v_{e,k}$ where $k=0$ if $s=s_e$ and $k=1$ if $s=1-s_e$.
Let $\pi \in \cWp$ with
\[\pi=\left(\left(\left((\digamma_{0,0},\digamma_{0,1}),\sigma_0\right),\ldots,\left((\digamma_{p-1,0},\digamma_{p-1,1}),\sigma_{p-1}\right)\right),\epsilon\right).
\]
Then one can check that for each $(i,s,v) \in [p/2] \times\lindexset$, we obtain $\pi(Q_{i,s,v})=C_{i,s,v}$, so that $\pi(\cQ)=\cP$, and thus $\cP$ is isomorphic to $\cQ$.
\end{proof}
\begin{lemma}\label{Henrietta}
Let $p$ be a positive even integer.  The set of all principal partitions in $\Part(p)$ is a $\cWp$-orbit in $\Part(p)$ of cardinality $(p-1)!! \, 8^{p/2}$.
\end{lemma}
\begin{proof}
Let $\cQ$ be the partition described in \cref{Horatio}.
The set of principal partitions is the $\cWp$-orbit of $\cQ$ by \cref{Horatio}.

We claim that the necessary and sufficient conditions on the permutation $\pi=((\delta_0,\ldots,\delta_{p-1}),\epsilon) \in \cWp$ for $\pi$ to be in $\Stab(\cQ)$ (the stabilizer of $\cQ$) are as follows:
\begin{itemize}
\item For every $e \in [p/2]$, we have $\epsilon(e+p/2) \equiv \epsilon(e)+p/2 \pmod{p}$.
\item For every $e \in [p/2]$, we have $\delta_e=\delta_{e+p/2}$.
\end{itemize}
If the first condition were violated, then there would be some $e \in [p/2]$ such that our permutation $\pi$ maps the class $\{(e,s,v),(e+p/2,s,v)\}$ of $\cQ$ to $\{(\epsilon(e),\delta_{\epsilon(e)}(s,v)), (\epsilon(e+p/2),\delta_{\epsilon(e+p/2)}(s,v)\}$ where $\epsilon(e+p/2)\not\equiv \epsilon(e)+p/2 \pmod{p}$, which means that our class of $\cQ$ is not being mapped to a class of $\cQ$.
If the first condition held but not the second, then there would be some $e \in [p/2]$ such that $\delta_e\not=\delta_{e+p/2}$, and so there will be some $s,v \in \lindexset$ such that $\pi$ maps the class $\{(\epsilon^{-1}(e),s,v),(\epsilon^{-1}(e)+p/2 \bmod{p},s,v)\}$ of $\cQ$ to $\{(e,\delta_e(s,v)), (e+p/2,\delta_{e+p/2}(s,v))\}$ where $\delta_e(s,v)\not=\delta_{e+p/2}(s,v)$, so that our class of $\cQ$ is not being mapped to a class of $\cQ$.
So our conditions are necessary.
Furthermore, if $\pi$ meets our two conditions, then for every $(e,s,v) \in [p/2]\times\lindexset$, the permutation $\pi$ maps the class $\{(e,s,v),(e+p/2,s,v)\}$ of $\cQ$ to $\{(\epsilon(e),\delta_{\epsilon(e)}(s,v)), (\epsilon(e+p/2),\delta_{\epsilon(e+p/2)}(s,v)\}=\{(e',s',v'),(e'+p/2 \bmod p,s',v')\}$ for some $e',s',v' \in [p/2]\times\lindexset$, i.e., $\pi$ maps every class of $\cQ$ to another class of $\cQ$.
Thus, our conditions are sufficient.

Now let us count how many $\pi=((\delta_0,\ldots,\delta_{p-1}),\epsilon) \in \cWp$ satisfy our necessary and sufficient conditions for $\pi$ to be in $\Stab(\cQ)$.
To count how many $\epsilon \in S_{[p]}$ satisfy the first condition, we introduce the terminology that two distinct elements $x, y \in [p]$ are {\it mates} if $x \equiv y+p/2 \pmod{p}$: so $[p]$ is a set partitioned into $p/2$ pairs of mates.
To meet the first condition, $\epsilon$ must map mates to mates.
We have $p$ choices for $\epsilon(0)$ (and $\epsilon(p/2)$ is forced to be the mate of $\epsilon(0)$).
This leaves $p-2$ choices for $\epsilon(1)$ (and $\epsilon(p/2+1)$ is forced to be the mate of $\epsilon(1)$).
Continuing in this manner, we see that there are $p!!$ possibilities for $\epsilon$.
To satisfy the second condition, we have  $8$ choices for each of $\delta_0,\ldots,\delta_{p/2-1}$ (since $S_2 \Wr_{[2]} S_2$ is isomorphic to the dihedral group of order $8$), and then $\delta_{p/2},\ldots,\delta_{p-1}$ must be the same list as the former list.
So $\card{\Stab(\cQ)}=p!! \, 8^{p/2}$.
Therefore, using \eqref{Samuel}, we see that the orbit of $\cQ$ (which is the set of all principal partitions in $\Part(p)$) has cardinality
\begin{align*}
\frac{\card{\cWp}}{\card{\Stab(\cQ)}}
& = \frac{p! \, 8^p}{p!! \, 8^{p/2}} \\
& = (p-1)!! \, 8^{p/2}. \qedhere
\end{align*}
\end{proof}

\subsection{Counting solutions}\label{Conrad}

In \cref{Ulrich}, we commented that if $p$ is a positive integer, $\cP\in\Part(p)$, and $\ell\in\N$, it is easier to enumerate $\As(\succeq\!\cP,=,\ell)$ than $\As(\cP,=,\ell)$.
Nonetheless, it still requires effort to compute $|\As(\succeq\!\cP,=,\ell)|$, so we introduce some notations and tools so that we can transform this problem into one of finding solutions to equations involving homogeneous linear polynomials.
We first give a definition that will supply the coefficients that appear in these polynomials.
\begin{definition}[Coefficient for $P$ in $e$, $\Co_{P,e}$]
For $P \subseteq \N \times \lindexset$ and $e \in \N$, we define the {\it coefficient for $P$ in $e$} to be
\[
\Co_{P,e} = |(\{e\} \times \{0\} \times [2]) \cap P| -|(\{e\} \times \{1\} \times [2]) \cap P|.
\]
\end{definition}
Many of the sets and partitions that we use have the following property, which is important in proving our results.
\begin{definition}[Effectual set and sets of sets]\label{Ellie}
We say that a set $P \subseteq \N \times \lindexset$ is {\it effectual} to mean that for every $e \in \N$ such that $(\eindexset) \cap P\not=\emptyset$ we have $\Co_{P,e}\not=0$.
We say that set $\cP$ of subsets of $\N\times\lindexset$ is {\it effectual} to mean that every $P \in \cP$ is effectual.
\end{definition}
We now introduce the variables that will occur in our linear polynomials.
\begin{notation}[$X_{\cP}$]
Let $p \in \N$ and $\cP \in \Part(p)$.
Then $X_{\cP}=\{X_P: P \in \cP\}$ is a set of $|\cP|$ distinct indeterminates indexed by the classes of $\cP$.
\end{notation}
Now we introduce our linear polynomials.
\begin{definition}[Form for $\cP$ in $e$, $\Fo_{\cP,e}$]\label{Doreen}
Let $p \in \N$ and $\cP\in \Part(p)$. Then for each $e \in [p]$, we define the {\it form for $\cP$ in $e$} to be
\begin{equation*}	  
\Fo_{\cP,e}(X_{\cP}) = \sum_{P \in \cP} \Co_{P,e} X_P.
\end{equation*}
\end{definition}
Notice that $\Fo_{\cP,e}$ is an element of the $|\cP|$-dimensional $\Q$-vector space of linear forms in the polynomial ring $\Q[X_{\cP}]$.
In our proofs, is important to distinguish partitions $\cP$ that have an associated form $\Fo_{\cP,e}$ in \eqref{Doreen} that is equal to $0$.
\begin{definition}[Degenerate partition]
Let $p \in \N$.  We say that a partition $\cP \in \Part(p)$ is {\it degenerate} to mean that there is some $e \in [p]$ such that $\Fo_{\cP,e}=0$.  Otherwise, $\cP$ is said to be {\it nondegenerate}.
\end{definition}
Effectuality, introduced in \cref{Ellie}, can be viewed as a stronger form of nondegeneracy.
\begin{lemma}\label{David}
Let $p \in \N$ and let $\cP \in \Part(p)$ be effectual.  Then $\cP$ is nondegenerate.
\end{lemma}
\begin{proof}
Let $e \in [p]$.  Since $\cP$ is a partition of $\indexset$, there is some $P \in \cP$ with $(e,0,0) \in P$.  Then $\Co_{P,e}\not=0$ because $\cP$ (and hence $P$) is effectual.  Thus $\Fo_{\cP,e}\not=0$.
\end{proof}
We now relate nondegeneracy and effectuality to our earlier concepts.
\begin{lemma}\label{Eva}
Let $p \in \N$ and $\cP \in \gelo(p)$.  Then $\cP$ is nondegenerate.
\end{lemma}
\begin{proof}
Let $e \in [p]$.  Since $\cP_{\{e\}}$ is not an even partition, there is some $P \in \cP$ such that $|(\eindexset) \cap P|$ is odd.  Therefore 
\begin{align*}
\Co_{P,e} & = |(\{e\} \times \{0\} \times [2]) \cap P| -|(\{e\} \times \{1\} \times [2]) \cap P| \\
& \equiv |(\{e\} \times \{0\} \times [2]) \cap P| + |(\{e\} \times \{1\} \times [2]) \cap P| \pmod{2} \\
& = |(\eindexset) \cap P| \\
& \equiv 1 \pmod{2},
\end{align*}
so that $\Co_{P,e} \not=0$, and thus $\Fo_{\cP,e}\not=0$.
\end{proof}
\begin{lemma}\label{Stanley}
If $p \in \N$ and $\cP \in \Con(p)$, then $\cP$ is effectual (and hence nondegenerate).
\end{lemma}
\begin{proof}
Let $P \in \cP$.
Let $e \in [p]$ with $(\eindexset)\cap P\not=\emptyset$.
By \cref{Epazote}, $P_{\{e\}}$ is of the form $\{(e,s,0),(e,s,1)\}$ for some $s \in [2]$, or else $P_{\{e\}}=\{(e,s,v)\}$ for some $s,v \in [2]$.  In the former case $|\Co_{P,e}|=2$, and in the latter case, $|\Co_{P,e}|=1$, so $\Co_{P,e}\not=0$.
This shows that every $P \in \cP$ is effectual, so that $\cP$ is effectual, and thus $\cP$ is also nondegenerate by \cref{David}.
\end{proof}
For a partition $\cP \in \Part(p)$, the linear forms of \cref{Doreen} provide an alternative way to think of the equations that an assignment $\tau$ must satisfy to be in $\As(\succeq\!\cP,=)$, where for each $P \in \cP$ the variable $X_P$ should be substituted with the common value of $\tau_{e,s,v}$ for all $(e,s,v) \in P$.  Rather than keep track of $\tau_{e,s,v}$ for every $(e,s,v) \in \indexset$ while enforcing the rule that $\tau_{e,s,v}=\tau_{e',s',v'}$ when $(e,s,v)$ and $(e',s',v')$ lie in the same class of $\cP$, we make an object that records one value for each class in $\cP$.
\begin{notation}[$\Fo_{\cP,e}(U)$]
Let $p,\ell \in \N$ and let $\cP\in\Part(p)$.
For any element $U \in [\ell]^{\cP}$, we write $\Fo_{\cP,e}(U)$ to mean the value one obtains by substituting, for each $P \in \cP$, the number $U(P)$ for the indeterminate $X_P$ in the polynomial $\Fo_{\cP,e}(X_{\cP})$.
\end{notation}
Now we can write $|\As(\succeq\!\cP,=,\ell)|$ in terms of our new concepts.
\begin{lemma}\label{Nancy}
Suppose that $p, \ell \in\N$ and $\cP \in \Part(p)$.  Then the set $\As(\succeq\!\cP,=,\ell)$ has the same cardinality as $\{U \in [\ell]^{\cP}: \Fo_{\cP,0}(U)=\cdots=\Fo_{\cP,p-1}(U)=0\}$.
\end{lemma}
\begin{proof}			
Let $\phi \colon \As(\succeq\!\cP,\ell) \to [\ell]^{\cP}$ be the map where for $\tau \in \As(\succeq\!\cP,\ell)$, we set $\phi(\tau) \in [\ell]^{\cP}$ such that for any $P \in \cP$ we have $(\phi(\tau))(P)$ equal to the common value of $\tau_{e,s,v}$ for every $(e,s,v) \in P$.
And let $\psi\colon [\ell]^{\cP} \to \As(\succeq\!\cP,\ell)$ be the map where for $U \in [\ell]^{\cP}$, we set $\psi(U) \in \As(\succeq\!\cP,\ell)$ such that for any $(e,s,v) \in \indexset$ we have $(\psi(U))_{e,s,v}=U(P)$ for whichever $P \in \cP$ contains $(e,s,v)$.  (This $P$ is exists and is unique because $\cP$ is a partition of $\indexset$.)
Then it is straightforward to check that $\phi$ and $\psi$ are inverses of each other, hence bijections.

For $\tau \in \As(\succeq\!\cP,\ell)$, we note that $\tau \in \As(\succeq\!\cP,=,\ell)$ if and only if $\Fo_{\cP,e}(\phi(\tau))=0$ for every $e \in [p]$.
Thus, if we restrict the domain of $\phi$ to $\As(\succeq\!\cP,=,\ell)$ and the codomain to $\{U \in [\ell]^{\cP}: \Fo_{\cP,0}(U)=\cdots=\Fo_{\cP,p-1}(U)=0\}$, then we obtain the bijection that establishes our claim.
\end{proof}
Eventually we shall show that $|\As(\cP,=,\ell)|$ is of the order of $\ell^{3 p/2}$ when $\cP$ is principal (\cref{Isaac}), but is of a lesser order when $\cP$ is any other contributory partition (\cref{Heinrich}).
Toward this end, we first prove that even, effectual, non-separable partitions have $|\As(\cP,=,\ell)|$ of a lesser order.
\begin{lemma}\label{Heidi}
Let $p$ be a positive integer, let $\cP \in \Part(p)$ be even, effectual, and not separable.
Then for each $\ell \in \N $
\[ |\As(\cP,=,\ell)|\leq |\As(\succeq\!\cP,=,\ell)| \leq \ell^{\ceil{3p/2}-1}.\]
\end{lemma}
\begin{proof}
The first inequality is due to \cref{Gertrude}, so it remains to prove the second.
Because $\cP$ is not separable, there is a pair $Q,R$ of equation-imbricate classes in $\cP$.
So there must be some $f,g \in [p]$ such that both $Q$ and $R$ intersect $\findexset$, but one and only one of $Q$ and $R$ intersects $\gindexset$.
By the effectuality of $\cP$, this means that $\Co_{Q,f}$ and $\Co_{R,f}$ are both nonzero, but exactly one of $\Co_{Q,g}$ or $\Co_{R,g}$ is nonzero.
So we may let $m \geq 2$ be the largest integer such that there are distinct classes $P_1, P_2, \dots , P_m \in \cP$ and distinct $e_1, e_2, \dots , e_m \in [p]$ such that
\begin{itemize}
\item $P_1$ and $P_2$ are equation-imbricate,
\item $\Co_{P_j,e_j} \not= 0$ for every integer $j$ with $1 \leq j \leq m$, and
\item $\Co_{P_j,e_k} = 0$ for every pair $j,k$ of integers with $1 \leq j < k \leq m$.
\end{itemize}
Now we let $P_{m+1}, \dots, P_{|\cP|} $ be the remaining classes of $\cP$, so that $\cP = \{ P_1, \dots, P_m, P_{m+1}, \dots, P_{|\cP|} \}$. 
For every $j$ with $1 \leq j \leq m$, we must have
\[ |\{ e \in [p]: \Co_{P_j,e}\not=0\}| \leq |P_j|,\]
since the former set is just a collection of some first coordinates of the triples in the latter.
Moreover, since $P_1$ and $P_2$ are equation-imbricate, we have
\[ | \{e \in [p] :  \Co_{P_1,e} \not=0 \} \cup \{ e \in [p]: \Co_{P_2,e}\not=0\}  | \leq  |P_1| + |P_2| -1 .\]
We claim that $\cap_{1 \leq j \leq m} \{ e \in [p] : \Co_{P_j,e} = 0 \}=\emptyset$: otherwise, we could name an element $h$ in this intersection, and then because $\cP$ is effectual, $(\hindexset) \cap P_j=\emptyset$ for all $j \in \{1,2,\ldots,m\}$.  Thus, there would be some class $T$ in $\{P_{m+1},\ldots,P_{|\cP|}\}$ such that $(\hindexset) \cap T\not=\emptyset$ and thus $\Co_{T,h}\not=0$, and then we would violate the maximality of $m$ using classes $P_1,\ldots,P_m,T$ and integers $e_1,\ldots,e_m,h$.
Therefore, we have $\cup_{1 \leq j \leq m} \{ e \in [p] : \Co_{P_j,e} \not= 0 \}=[p]$, so we arrive at the following inequality
\begin{align*}
p
&=  \Bigg|\bigcup_{1 \leq j \leq m}  \{ e \in [p]: \Co_{P_j,e}\not=0\}  \Bigg|\\[3pt]
& \leq |\{e \in [p] :  \Co_{P_1,e} \not=0 \} \cup \{ e \in [p]: \Co_{P_2,e}\not=0\} |  \\ & \qquad\qquad\qquad\qquad\qquad\qquad\qquad\qquad\qquad +  \sums{2< j \leq m} | \{ e \in [p]: \Co_{P_j,e} \not=0 \} | \\
&\leq |P_1| + |P_2| -1 + |P_3| + \dots + |P_m| \\
&< |P_1| + |P_2|+ |P_3| + \dots + |P_m|.
\end{align*}
Since $\cP $ is a partition of $ \indexset $ then we also know that 
\[  |P_1| + \dots + |P_m| + |P_{m+1}| + \dots + |P_{|\cP|}| = 4p. \]
Hence we can conclude that $|P_{m+1}| + \dots + |P_{|\cP|}| < 3p$.
Since each class in the partition $\cP$ is nonempty and of even size, we must have $|\cP|-m \leq \ceil{3 p/2}-1$.

Now $\Fo_{\cP,e_1}, \dots, \Fo_{\cP,e_m}$ is a $\Q$-linearly independent list of polynomials, since for each $j \in \{1,2,\ldots,m\}$ the form $\Fo_{\cP,e_j}$ has a nonzero coefficient $\Co_{P_j,e_j}$ for the variable $X_{P_j}$ but $X_{P_j}$ has a zero coefficient $\Co_{P_j,e_k}$ in the forms $\Fo_{\cP,e_k}$ with $j < k \leq m$.
Thus, 
\begin{align*}
m &= \dim_{\Q} \Span_{\Q} (\Fo_{\cP,e_1}, \dots, \Fo_{\cP,e_m} ) \\
& \leq  \dim_{\Q} \Span_{\Q} (\Fo_{\cP,0}, \dots, \Fo_{\cP,p-1}).
\end{align*}
This implies that  
\[ |\cP|-\dim_{\Q}\Span_\Q(\Fo_{\cP,0},\ldots ,\Fo_{\cP,p-1})  \leq |\cP|-m \leq \ceil{3 p/2}-1.  \]
\cref{Nancy} tells us that $|\As(\succeq\!\cP,=,\ell)|$ is the same as 
\[
|\{U \in [\ell]^{\cP}: \Fo_{\cP,0}(U)=\cdots=\Fo_{\cP,p-1}(U)=0\}|,
\]
and our last inequality tells us that the number of free variables in the system of equations used to define the last set is bounded above by $\ceil{3 p/2}-1$.  Since each free variable must be assigned a value in $[\ell]$, the cardinality of this set is at most $\ell^{\ceil{3 p/2}-1}$.
\end{proof}
Now we upper bound $|\As(\cP,=,\ell)|$ when $\cP$ is nondegenerate and separable, and then use this in \cref{Heinrich} to show that $|\As(\cP,=,\ell)|$ is of order less than $\ell^{3 p/2}$ when $\cP$ is a contributory partition that is non-principal.
\begin{lemma} \label{cookies}
Let $p$ be a positive integer, let $\cP \in \Sep(p)$ be nondegenerate, and let $\fP$ be the quotient of $\cP$. For each $\ell \in \N $ we have 
\[
|\As(\cP,=,\ell)| \leq   |\As(\succeq\!\cP,=,\ell)| \leq \ell^{|\cP|-|\fP|}.
\]
In particular, this bound holds whenever $\cP \in \gelo(p)\cap\Sep(p)$ and (therefore also whenever $\cP \in \Con(p)\cap\Sep(p)$).
\end{lemma}
\begin{proof}
The last statement about partitions in $\gelo(p)\cap\Sep(p)$ and in $\Con(p)\cap\Sep(p)$ will follow from the earlier claims because every GELO partition is nondegenerate by \cref{Eva}, and every contributory partition is GELO.

The first claimed inequality is due to \cref{Gertrude}, so it remains to prove the second.
\cref{Nancy} tells us that $|\As(\succeq\!\cP,=,\ell)|=|\{U \in [\ell]^{\cP}: \Fo_{\cP,0}(U)=\cdots=\Fo_{\cP,p-1}(U)=0\}|$.
From linear algebra, the number of solutions $U \in [\ell]^{\cP}$ to a system of homogeneous linear equations $\Fo_{\cP,0}=\cdots=\Fo_{\cP,p-1}=0$ in $|\cP|$ variables is at most $\ell^{|\cP|-\dim_{\Q}\Span_\Q(\Fo_{\cP,0},\ldots,\Fo_{\cP,p-1})}$.  So it suffices to prove that $\dim_{\Q}\Span_\Q(\Fo_{\cP,0},\ldots,\Fo_{\cP,p-1}) \geq |\fP|$. 

For each $\cC \in \fP$, let $E_{\cC}\subseteq [p]$ be the set of all $e \in [p]$ such that, for every $P \in \cC$, we have $(\eindexset)\cap P\not=\emptyset$.
Because $\cP$ is separable, $\{E_{\cC}: \cC \in \fP\}$ is a partition of $[p]$.
If $\cC$ and $\cD$ are two different elements of $\fP$, then for each $Q \in \cD$, and each $e \in E_{\cC}$, we have $(\eindexset)\cap Q=\emptyset$, and thus $\Co_{Q,e}=0$.  
So the linear forms in $\{\Fo_{\cP,e}: e \in E_{\cC}\}$ can only involve the variables in $\{X_P: P \in \cC\}$, and these linear forms are nonzero by the nondegeneracy of $\cP$.
Thus, $\Span_\Q \{\Fo_{\cP,e}: e \in [p]\}$ is the internal direct sum of the spaces $\Span_\Q \{\Fo_{\cP,e}: e \in E_{\cC}\}$ as $\cC$ runs through $\fP$, and so
\begin{align*}
\dim_\Q \Span_\Q \{\Fo_{\cP,e}: e \in [p]\} 
& = \sum_{\cC \in \fP} \dim_\Q \Span_\Q \{\Fo_{\cP,e}: e \in E_{\cC}\} \\
& \geq |\fP|,
\end{align*}
where the last inequality uses the fact that for each $\cC \in \fP$, the forms in the set $\{\Fo_{\cP,e}: e \in E_{\cC}\}$ are nonzero, and there is at least one form in the set since $E_{\cC}$ cannot be empty.
\end{proof}
\begin{lemma}\label{Heinrich}
Let $p \in \N$ and suppose that $\cP \in \Con(p)$.
If $\cP$ is not principal then
\[
|\As(\cP,=,\ell)| \leq   |\As(\succeq\!\cP,=,\ell)| \leq  \ell^{\ceil{ 3p/2} -1}.
\]
\end{lemma}
\begin{proof}
If $\cP\not\in\Sep(p)$, then since $\cP \in \Con(p)$, it is even because $\Con(p) \subseteq \gelo(p)$ and $\cP$ is effectual by \cref{Stanley}, so then the result is immediate from \cref{Heidi}.
If $\cP\in\Sep(p)$, the result is immediate from \cref{cookies} \and \cref{milk}.
\end{proof}
Now we prove that $|\As(\succeq\!\cP,=,\ell)|$ is of order $\ell^{3 p/2}$ when $\cP$ is principal; this helps us toward our ultimate goal of proving that $|\As(\cP,=,\ell)|$ is of that order.
\begin{lemma}\label{Gracie}
If $\cP \in \Part(p)$ is a principal partition then
\[ \card{\As(\succeq\!\cP,=,\ell)} = \left(\frac{2\ell^3+\ell}{3}\right)^{p/2}. \]
\end{lemma}
\begin{proof}
Let the partition $\cQ$ and the classes $Q_{e,s,v}$ for $(e,s,v)\in [p/2]\times\lindexset$ be as in \cref{Horatio}, which (along with \cref{Esmeralda}) shows that it suffices for us to prove the desired inequality when $\cP=\cQ$.

By \cref{Nancy}, it then suffices to show that
\[
|\{U \in [\ell]^{\cP}: \Fo_{\cQ,0}(U)=\cdots=\Fo_{\cQ,p-1}(U)=0\}|= \left(\frac{2\ell^3+\ell}{3}\right)^{p/2}. \]
Notice that for $f \in [p]$, we have
\begin{align*}
\Fo_{\cQ,f}
& = \sum_{Q \in \cQ} \Co_{Q,f} X_Q \\
& = \sum_{(e,s,v) \in [p/2]\times\lindexset} \Co_{Q_{e,s,v},f} X_{Q_{e,s,v}} \\
& = \sum_{(e,s,v) \in \{f \bmod{p/2}\} \times\lindexset} \Co_{Q_{e,s,v},f} X_{Q_{e,s,v}} \\
& =  X_{Q_{f\bmod{p/2},0,0}} + X_{Q_{f\bmod{p/2},0,1}} - X_{Q_{f\bmod{p/2},1,0}} - X_{Q_{f\bmod{p/2},1,1}}.
\end{align*}
In particular, notice that if $f \in [p/2]$ then $\Fo_{\cQ,f}=\Fo_{\cQ,f+p/2}$.
Therefore, it suffices to show that
\[
|\{U \in [\ell]^{\cP}: \Fo_{\cQ,0}(U)=\cdots=\Fo_{\cQ,p/2-1}(U)=0\}|= \left(\frac{2\ell^3+\ell}{3}\right)^{p/2}.
\]
Furthermore, notice that if $f$ and $g$ are distinct elements of $[p/2]$, then $\Fo_{\cQ,f}$ and $\Fo_{\cQ,g}$ have no variable in common.
Therefore, the number of solutions $U \in [\ell]^\cP$ to the system of equations $\Fo_{\cQ,0}(U)=\cdots=\Fo_{\cQ,p/2-1}(U)$ is equal to the product of the numbers of solutions for the individual equations.
Each equation has $(2\ell^3+\ell)/3$ solutions by \cref{wowzers}.
\end{proof}
The following result helps us account for the difference between $|\As(\succeq\!\cP,=,\ell)|$ and $|\As(\cP,=,\ell)|$ when $\cP$ is principal, so that we may convert the result in \cref{Gracie} on $|\As(\succeq\!\cP,=,\ell)|$ to a result on $|\As(\cP,=,\ell)|$ in \cref{Isaac}.

\begin{lemma}\label{Rutilia}
Suppose that $p \in \N$ and $\cP\in\Part(p)$ is a principal partition. If $\cR\succ\cP$ then
\[
|\As(\cR,=,\ell)| \leq |\As(\succeq\!\cR,=,\ell)| \leq \ell^{\ceil{ 3p/2} -1}.
\]
\end{lemma}
\begin{proof}
The first inequality is due to \cref{Gertrude}, so it remains to prove the second.
Let $\cQ$ be a partition with $\cR \succeq \cQ \succ \cP$ and such that there is no $\cV$ with $\cQ \succ \cV \succ \cP$.
Because $\cR \succeq \cQ$, by \cref{Julian} it suffices to show that 
\begin{equation}\label{Julie}
|\As(\succeq\!\cQ,=,\ell)| \leq \ell^{\ceil{3p/2} -1}.
\end{equation}

Because there is no $\cV$ with $\cQ \succ \cV \succ \cP$, there must be exactly one class $U$ in $\cQ$ that is the union of two distinct classes $S, T$ in $\cP$, while the remaining classes of $\cQ$ are the remaining classes in $\cP$.
Let $\fP$ be the quotient of $\cP$.
Let $\cS$ and $\cT$ be the classes of $\fP$ such that $S \in \cS$ and $T \in \cT$.

Let us first examine the case when $\cS\not=\cT$.
Then the class $U$ is the union of the two equation-disjoint classes $S$ and $T$, and so $U$ is equation-imbricate with $S$ and with $T$, and indeed with all the classes in $\cS$ and $\cT$.
Since $\cS$ and $\cT$ have four classes each by \cref{water}\ref{broccoli}, three classes from each of $\cS$ and $\cT$ are classes of $\cQ$, and so $\cQ$ is not separable.
Now $\cP$ is principal, hence contributory, and so is effectual by \cref{Stanley}.
Thus all the classes of $\cQ$ that are classes of $\cP$ are effectual.
Furthermore, since $U=S\cup T$ and $S$ and $T$ are equation-disjoint, we see that for every $e \in [p]$ such that $(\eindexset)\cap U\not=\emptyset$, we either have $(\eindexset)\cap S\not=\emptyset$ or $(\eindexset)\cap T\not=\emptyset$ but never both, and since $S$ and $T$ are effectual, this means that either $\Co_{U,e}=\Co_{S,e}\not=0$ or $\Co_{U,e}=\Co_{T,e}\not=0$.
Thus, $U$ is effectual, and so $\cQ$ is effectual.
Also, $\cP$ is even since it is principal, hence contributory, hence GELO, and so $\cQ$ is even since each class of $\cQ$ is a class or union of classes of $\cP$.
Therefore, \eqref{Julie} follows from \cref{Heidi}.

Now let us examine the case when $\cS=\cT$.
Then $S$ and $T$ are equation-identical to each other and to $U=S\cup T$.
This means that $\cQ$ is a separable partition.
We let $\fQ$ be the quotient of $\cQ$.
We set $\cU=(\cS\smallsetminus\{S,T\})\cup \{U\}$, and then $\fQ=(\fP\smallsetminus\{\cS\})\cup\{\cU\}$.
In particular, $|\fQ|=|\fP|$, and since $|\cQ|=|\cP|-1$, this means
\begin{align}
\begin{split}\label{Eunice}
|\cQ|-|\fQ|
& = |\cP|-|\fP|-1 \\
& =\frac{3 p}{2}-1  \\
& =\left\lceil\frac{3 p}{2}\right\rceil-1
\end{split}
\end{align}
where the penultimate equality uses \cref{water}\ref{tomato} and the ultimate uses the fact that $p$ must be even by \cref{milk}.
Also since every class of $\fP$ has cardinality $4$ by \cref{water}\ref{broccoli}, we see that every class of $\fQ$ has cardinality $3$ (in the case of $\cU$) or $4$ (for the rest).
This means that for every $e \in [p]$, the partition $\cQ_{\{e\}}$ of $\eindexset$ has three or four classes, and hence has a singleton class, and so there is some class $Q \in \cQ$ such that $|\Co_{Q,e}|=1$, thus making $f_{\cQ,e}\not=0$.
Thus, $\cQ$ is nondegenerate.
Therefore, \eqref{Julie} follows from \cref{cookies} and \eqref{Eunice}.
\end{proof}
Now we can finally prove that $|\As(\cP,=,\ell)|$ is of order $\ell^{3 p/2}$ when $\cP$ is principal.
\begin{lemma}\label{Isaac}
If $p \in \N$ and $\cP\in\Part(p)$ is principal then 
\[\lim_{\ell \to \infty}  \frac{|\As(\cP,=,\ell)|}{\ell^{3p/2}} =  \left(\frac{2}{3}\right)^{p/2}.\]
\end{lemma}	
\begin{proof}
By \cref{Gertrude} we have
\[
\card{\As(\cP,=,\ell)}=\card{\As(\succeq\!\cP,=,\ell)} - \sum_{\cR \succ \cP} \card{\As(\cR,=,\ell)}.
\]
Now divide both sides by $\ell^{3 p/2}$ and take the limit as $\ell$ tends to infinity.
Because of Lemmas \ref{Gracie} and \ref{Rutilia}, the limit of the right hand side is $(2/3)^{p/2}$.
\end{proof}

\subsection{Asymptotic moments}\label{Abraham}

Now we know enough about the asymptotic behavior of $\As(\cP,=,\ell)$ as $\ell\to\infty$ for partitions $\cP\in\Con(p)$ to employ \cref{Sanria} to obtain the limiting behavior of the $p$th central moment of $\ssac$.
\begin{theorem}\label{Tatiana}
If $p \in \N$, then  
\[
\lim_{\ell \to \infty}
\frac{\mom{p} \ssac(f)}{\ell^{3 p/2}} =
\begin{cases}
0 & \text{if $p$ is odd,} \\
(p-1)!! \,(16/3)^{p/2} & \text{if $p$ is even,}
\end{cases}
\]
\end{theorem}
so that
\[
\lim_{\ell \to \infty} \ell^{p/2}
\mom{p} \ADF(f) =
\begin{cases}
0 & \text{if $p$ is odd,} \\
(p-1)!! \,(16/3)^{p/2} & \text{if $p$ is even.}
\end{cases}
\]
\begin{proof}
By \cref{Sanria}
\[
\mom{p} \ssac(f) = \sum_{\fP \in \Isom(p)} |\fP| \Sols(\fP,\ell).
\]
If $\fP$ is an isomorphism class containing non-principal partitions, then
\[
\lim_{\ell\to\infty} \frac{|\Sols(\fP,\ell)|}{\ell^{3 p/2}} =0
\]
by \cref{Heinrich}.
Since \cref{milk} tells us that there are no principal partitions when $p$ is odd, this proves our result about $\ssac$ in this case.
On the other hand, if $p$ is positive and even, then \cref{Henrietta} tells us that there is a single isomorphism class $\fQ$ containing all $(p-1)!! 8^{p/2}$ principal partitions, and \cref{Isaac} tells us that
\[
\lim_{\ell\to\infty} \frac{|\Sols(\fQ,\ell)|}{\ell^{3 p/2}} =\left(\frac{2}{3}\right)^{p/2}.
\]
Thus, if we divide $\mom{p} \ssac(f)$ by $\ell^{3 p/2}$ and take the limit as $\ell$ tends to $\infty$, we obtain $(p-1)!! 8^{p/2} (2/3)^{p/2}$.
Finally, if $p=0$, then the $p$th moment $\mom{p} \ssac(f)$ is trivially $1$ and $(-1)!!=1$, so the desired result also holds.
The result for $\ADF$ follows from that for $\ssac$ by equation \eqref{Natasha} from the Introduction.
\end{proof}
We can now use the variance of the demerit factor from \cref{Albert} to obtain the limiting standardized moments of the demerit factor.
\begin{corollary}\label{Indy}
Let $p$ be a nonnegative integer.  Then
\[
\lim_{\ell \to \infty} \smom{p} \ADF(f) = \lim_{\ell \to \infty} \smom{p} \ssac(f) =
\begin{cases}
0 & \text{if $p$ is odd,} \\
(p-1)!! & \text{if $p$ is even.}
\end{cases}
\]
\end{corollary}
\begin{proof}
The limit for the standardized moment of $\ADF(f)$ follows from \cref{Tatiana} and \cref{Albert}, and since $\ADF(f)=-1+\ssac(f)/\ell^2$, the standardized moments of $\ssac(f)$ are the same as those of $\ADF(f)$.
\end{proof}
\begin{remark}
The $p$th central moment of a standard normal distribution is $0$ if $p$ is odd and $(p-1)!!$ if $p$ is even.
We see that the limiting normalized central moments in \cref{Indy} have the same values.
\end{remark}

\section*{Acknowledgement}

The authors thank Bernardo \'Abrego and Silvia Fern\'andez-Merchant for helpful discussions and suggestions.

\end{document}